\documentclass[journal,12pt,onecolumn,draftclsnofoot,]{IEEEtran}
\IEEEoverridecommandlockouts
\usepackage{rotating}
\usepackage{epstopdf}
\usepackage{cite}
\usepackage{float}

\usepackage[utf8]{inputenc} 
\usepackage[T1]{fontenc}
\usepackage{url}
\usepackage{ifthen}
\usepackage[switch]{lineno}
\usepackage{cite}
\usepackage[cmex10]{amsmath} 
\usepackage{amssymb}
\usepackage{graphicx}
\usepackage{helvet}
\usepackage{titlesec}
\usepackage{epstopdf}
\usepackage{epsfig}
\usepackage{titlesec}
\usepackage{multirow}
\usepackage{amsthm}
\usepackage{comment}
\usepackage{caption}
\usepackage{amssymb}
\usepackage[cmex10]{amsmath}
\usepackage{breqn}
\usepackage{graphicx}
\usepackage{stfloats,adjustbox}
\usepackage{url}
\usepackage{xcolor}
\usepackage{multirow}
\usepackage{bm}
\usepackage{bbm}
\usepackage{amsmath}
\usepackage{enumitem}
\usepackage{subcaption}
\usepackage[english]{babel}
\usepackage{ragged2e}
\usepackage{blindtext}
\usepackage{caption, multirow, makecell}
\newcommand{\subscript}[2]{$#1 _ #2$}

\usepackage{fixltx2e}
\usepackage{longtable,eqnarray,makecell}
\usepackage{amsmath,amssymb}
\usepackage{threeparttable,stackengine}
\usepackage{tikz}
\usepackage{amsthm}
\newtheorem{example}{Example}

\newtheorem{note}{Note}

\usepackage{booktabs}
\newtheorem{thm}{Theorem}
\newtheorem{lem}{Lemma}

\newtheorem{cor}{Corollary}

\begin{document}
\title{Coded Caching for Hierarchical Two-Layer Networks with Coded Placement\\
} 

\author{\IEEEauthorblockN{Rajlaxmi~Pandey, Charul Rajput, and B. Sundar Rajan}\\
	\IEEEauthorblockA{\textit{Department of Electrical Communication Engineering} \\
		\textit{Indian Institute of Science, Bangalore, India} \\
		\{rajlaxmip,charulrajput,bsrajan\}@iisc.ac.in}
}

\maketitle

\begin{abstract} 

We examine a two-layered hierarchical coded caching problem, a configuration addressed in existing research. This involves a server connected to $K_1$ mirrors, each of which serves $K_2$ users. The mirrors and the users are equipped with caches of size $M_1$ and $M_2$, respectively. We propose a hierarchical coded caching scheme with coded placements that outperforms existing schemes. To ensure a fair comparison, we introduce the notion of composite rate, defined as $\overline{R} = R_1 + K_1 R_2$, where $R_1$ is the rate from the server to mirrors and $R_2$ is the rate from mirrors to users. The composite rate has not been discussed before in the literature and is pertinent when mirrors transmit with different carrier frequencies. For the proposed scheme, we show a trade-off between the global memory $\overline{M}=K_1M_1+K_1K_2M_2$ of the system and the composite rate and compare with the existing schemes. Additionally, we conduct this comparative analysis by plotting $R_1$ + $R_2$ against global memory, which is particularly beneficial for systems wherein each mirror can utilize the same carrier frequency, given their significant spatial separation.  Additionally, we propose an optimized scheme for the specific case of a single mirror, showing improved performance in this scenario.

\end{abstract} 
\IEEEpeerreviewmaketitle

\section{Introduction and Background}
The world has witnessed a huge exponential growth in data demand, especially in the last decade, which is projected to grow more abruptly in the near future owing to the increased dependence on smartphones, electronic gadgets, etc. Video-on-demand is the biggest contributor to data traffic. It is observed that the traffic is not uniform throughout the day. Coded caching has been proposed as a promising technique to mitigate data traffic during peak hours. The seminal work of \cite{maddah2014fundamental} proposed combined placement and delivery phases to characterize both the cache placement and transmission. In the placement phase, which mostly takes place during off-peak traffic hours, contents are put in the user's cache by the central server. The transmission of the contents of the demanded files not cached by the users takes place during the delivery phase. The "transmission rate" refers to the number of coded symbols the server needs to transmit to fulfill the demands of all the users.

The set-up of \cite{maddah2014fundamental} consists of a central server having access to a library of \emph{N} files, which is connected to \emph{K} users through an error-free shared link.
Each user has a cache size of \emph{M} files. For this set-up, a ({\emph{K}, \emph{M}, \emph{N}) coded caching scheme is given in \cite{maddah2014fundamental}, which we refer to as MN (Maddah-Ali, Niesen) scheme henceforth. For placement, each file is divided into \emph{F} packets, and only up to \emph{MF} packets can be stored in each user's cache. \emph{F} is termed as the sub-packetization level. 
In \cite{maddah2014decentralized}, a coded caching scheme (MN decentralized scheme) was proposed to address the scenarios when a central server is unavailable, and the users must place the cache contents themselves. As the number of users increases, the sub-packetization level $F$ increases exponentially, making the original MN scheme \cite{maddah2014fundamental} unfeasible. Hence, to reduce the sub-packetization level, in \cite{yan2017placement}, the authors proposed a scheme based on a placement delivery array (PDA) to characterize both the placement and delivery phases simultaneously.
Coded caching has been studied in a variety of settings, such as users with non-uniform demands \cite{niesen2016coded,ji2017order}, multi-access networks \cite{hachem2017coded}, combinatorial multi-access caching\cite{brunero2022fundamental}, shared cache networks \cite{parrinello2019fundamental}, Multi-server networks \cite{shariatpanahi2016multi}, users with asynchronous demands \cite{lampiris2021coded}, land mobile satellite systems\cite{zhao2022coded}, online caching\cite{pedarsani2015online}, hierarchical networks\cite{karamchandani2016hierarchical,zhang2016decentralized,wang2019reduce,liu2021intension,kong2022centralized}.

In practice, the caching systems comprise multiple layers between the central server and end-users. The in-between devices act as mirrors for connectivity between various layers.
In \cite{karamchandani2016hierarchical}, the authors considered a setting where the server is connected through an error-free broadcast link to $K_{1}$ mirror sites. Further, each
mirror connects to $K_{2}$ users through an error-free broadcast link. As a result, the total number of users $K$ in the network is equal to the product of $K_{1}$ and $K_{2}$, i.e., $K=K_{1}K_{2}$.
Let $R_{1}$ and $R_{2}$ denote the transmitted rate (normalized by the file size) over the shared link from the server to the mirrors and mirrors to the users, respectively, to satisfy the demand of the users. For this setting, a coded caching scheme based on decentralized placement was proposed. Later in \cite{kong2022centralized}, a hierarchical coded caching scheme was proposed for the centralized server, which is based on the MN scheme.

\subsection{System model}
A ($\emph{$K_{1}$},\emph{$K_{2}$};\emph {$M_{1}$},\emph{$M_{2}$};\emph{N}$) hierarchical coded caching is illustrated in Fig. \ref{fig:setting}. The server has a collection of $N$ files, each of size $F$ bits denoted by $W_1,W_2,\hdots,W_N$. There are $\emph{$K_{1}$}$ mirrors, and to each mirror, $\emph{$K_{2}$}$ users are connected through an error-free shared link. Each mirror and user has a cache memory of size ${M_1}$  and ${M_2}$ files, respectively. 

The ($\emph{$K_{1}$},\emph{$K_{2}$};\emph{$M_{1}$},\emph{$M_{2}$};\emph{N}$) hierarchical coded caching problem operates in the following two phases:

\begin{enumerate}[label=(\alph*)]
\item \emph{Placement phase}: In this phase, the server places the contents of files in the mirrors' and users' caches without the knowledge of the users' demands.
\item \emph{Delivery phase}:
\begin{enumerate}[label=(\roman*)]
\item \emph{Delivery by server}: Users reveal their demands to the server. Then, the server broadcasts coded or uncoded multi-cast messages of size at most $R_1F$ bits to each mirror site.
\item \emph{Delivery by mirror}: Each mirror generates a coded or uncoded message of size at most $R_2F$ bits based on the transmission by the server and its cache content and then sends them to the attached users.
\end{enumerate}
\end{enumerate}

The \textit{global memory} of the system is defined as $\overline{M}=K_1M_1+K_1K_2M_2$, and the \textit{composite rate} is defined as $\overline{R}=R_1+K_1R_2$.
\begin{figure}[t!]
	\begin{center}
		\captionsetup{justification=centering}
		\includegraphics[width=0.7\textwidth]{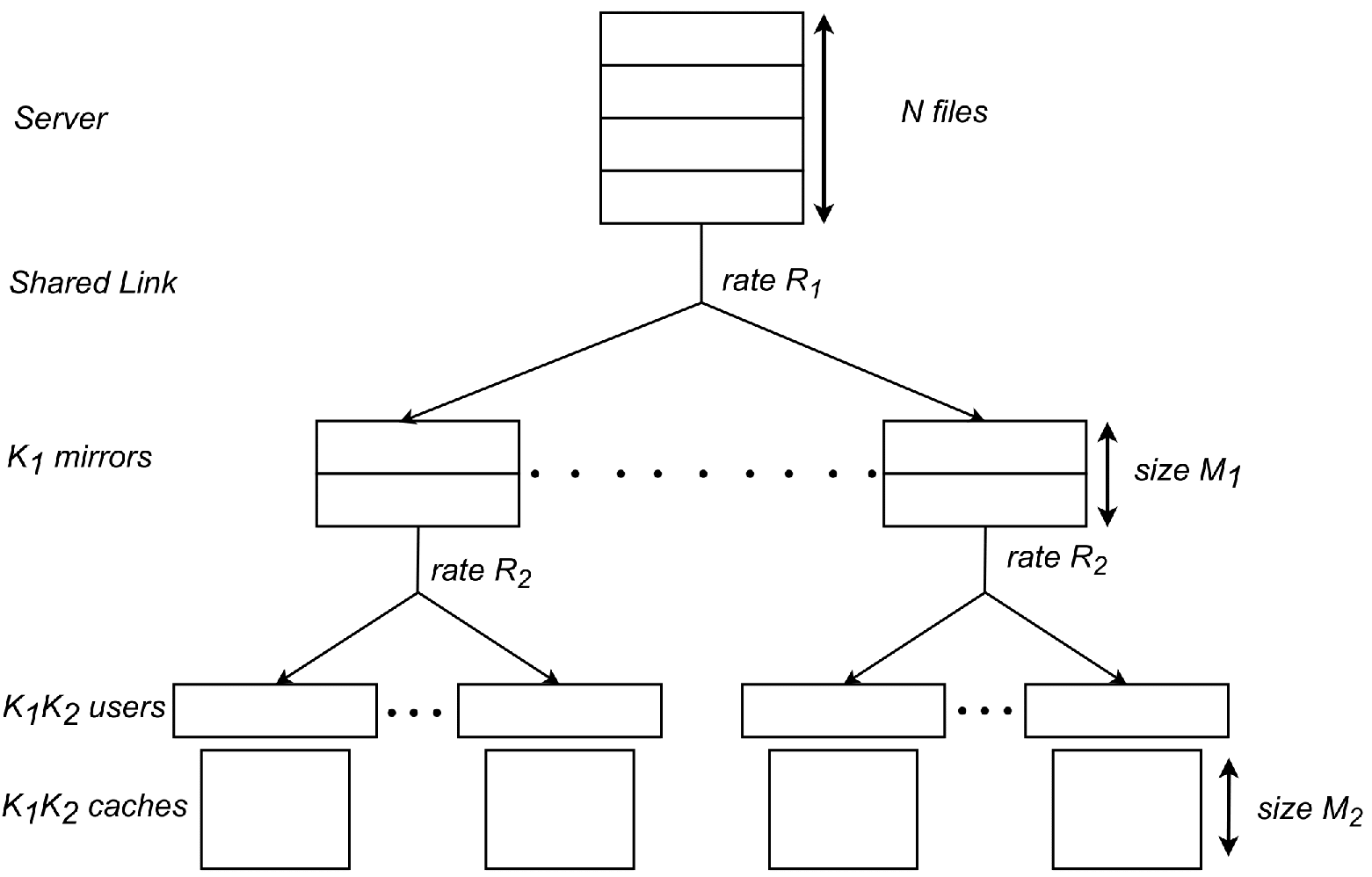}
		\caption{Hierarchical two-layer network.}
		\label{fig:setting}
	\end{center}
\end{figure}

\subsection{Prior Work}

Hierarchical coded caching was first studied in \cite{karamchandani2016hierarchical}, and we refer to the scheme as the KNMD (Karamchandani, Niesen, Maddah-Ali, Diggavi) scheme. The KNMD scheme divides a file sharing system into two subsystems based on fixed parameters $\alpha$ and $\beta$. The first subsystem includes mirrors' cache memory and users' cache memory, responsible for caching and delivering the $\alpha$  parts of each file. The second subsystem includes users' cache memory only, responsible for caching and delivering the $(1-\alpha)$ parts of each file. A $\beta$ fraction of each user's cache memory is used in the first subsystem, and $(1-\beta)$ fraction is used in the second subsystem. Both subsystems use the same single-layer decentralized coded caching scheme with different parameters to optimize efficiency. 

In the KNMD scheme, the server ignores the users' cache contents when sending messages to mirror sites, which can cause redundant communication costs in the first layer.
In \cite{8481555}, referred to as the ZWXWLL (Zhang, Wang, Xiao, Wu, Liang, Li) scheme, the correlations between the prestored contents in two layers of a decentralized caching system are utilized to optimize data placement and reduce the transmission rate in the first layer while maintaining the second layer's transmission rate through an optimized delivery strategy. 
In another work \cite{7763243}, the authors introduced a centralized scheme for this hierarchical two-layer network where they proposed a joint caching scheme in order to consider the caches of both layers, thus eliminating unnecessary transmission. We refer to this scheme as the ZWXWL (Zhang, Wang, Xiao, Wu, Li) scheme.

 Following this, in \cite{wang2019reduce}, the authors created a caching scheme, referred to as the WWCY (Wang, Wu, Chen, Yin) scheme that utilizes idle time resources by allowing simultaneous transmission between the server and relays. 
The WWCY scheme exhibits a pattern where authors recommend the same memory ranges and optimal values of $\alpha$ and $\beta$ for each region, similar to the KNMD scheme. 
 Notably, the WWCY scheme achieves the minimum load of $R_{2}$ when employing uncoded placement, and furthermore, the load of the first layer $R_{1}$  is lower than that of the ZWXWLL scheme.

 In \cite{liu2021intension}, referred to as the LZX (Liu, Zhang, Xie) scheme, the tension between the rates of different layers was studied using a toy model and derived new lower bounds and achievable schemes. They considered the case of a single mirror with two users attached. 

A new combination structure referred to as hierarchical placement delivery array (HPDA), inspired by PDA, was proposed in \cite{kong2022centralized}, referred to as the KWC (Kong, Wu, Cheng) scheme, which could be used to characterize the placement and delivery phase of a hierarchical coded caching scheme.

\subsection{Our Contributions}

The proposed scheme is inspired by the scheme in \cite{chen2016fundamental}, which we refer to as the CFL (Chen, Fan, Letaief) scheme henceforth and MN \cite{maddah2014fundamental} coded caching scheme. We divide each file into two sub-files of sizes $F_1$ and $F_2$ bits. For the first sub-file, we are using a variation of the CFL scheme with sub-packetization level $K$, and for the second sub-file, we are using the MN scheme with $t\in \{1,2,\cdots, K_{1}K_{2}\}$. Since the CFL scheme gives a better rate when $N \leq K$, the proposed scheme also performs better for the case when $N \leq K_1K_2$. We compare our scheme with the state-of-the-art with a similar setting.
Our contributions are summarized as follows:
\begin{itemize}
 \item In the hierarchical coded caching setup, we utilize coded placement techniques inspired by the CFL scheme and propose a new scheme. Our results demonstrate a notable reduction in the required rates for both $R_{1}$ and $R_{2}$. 
       
 \item In order to ensure a fair comparison with existing schemes, we propose the notion of composite rate, $\overline{R}$, and the global memory, $\overline{M}$, and we utilize these parameters as our primary metrics for evaluating the global bandwidth requirements of the schemes. The global memory parameter, $\overline{M}$, reflects the overall memory capacity of the system, while the composite rate, $\overline{R}$, quantifies the total bandwidth present in the system.

\item We compare the performance of our proposed scheme with the state-of-the-art scheme by plotting curves of their respective composite rates against the available global memory. The resulting curve is presented in Figure \ref{fig:CRvsGM}.  Additionally, we examine the trade-off for \( R_1 + R_2 \) with global memory, as illustrated in Figure \ref{R_1+R_2}. We also analyze the trade-offs for \( R_1 \) and \( R_2 \) separately with \( M_1 \) and \( M_2 \), in comparison to existing schemes.

\item  We demonstrate how memory sharing works within our scheme by leveraging the memory of both the mirror's cache ($M_1$) and the users' cache ($M_2$) using a convex hull approach.

\item For the case of only one mirror in the system, we improve our scheme and compare that with the existing scheme in the literature.
       
\end{itemize}

\subsection{Organization}
The rest of the paper is organized as follows: The main results containing the proposed scheme and performance analysis are listed in Section \ref{scheme1}. In Section \ref{compare}, we present an in-depth comparison with the current state of the art. Section \ref{scheme2} introduces an alternative scheme for cases involving a single mirror. Section \ref{scheme2compare} provides a comparative analysis with existing literature for the single mirror case. Concluding remarks are included in Section \ref{conclu}.

\subsection{Notations}

For a positive integer $n$, $[n]$ denotes the set $\{1,2,\cdots,n\}$. For two sets $A$ and $B$, the notation $A \setminus B$ denotes the set of elements in A that are not in B. For a set $A$, the number of elements in it is represented by $|A|$. The binomial coefficient represented as $\binom{n}{k}$ is equal to $\frac{n!}{k!(n-k)!}$, where $n$ and $k$ are positive integers such that $k \leq n$. For a negative integer $k$, the value of $\binom{n}{k}$ is defined as zero. For a set $S$ and a positive integer $t$, the notation ${S \choose t}$ represents the collection of all the subsets of $S$ of size $t$.

\section{Proposed Coded Caching Scheme for Hierarchical Network} \label{scheme1}
 Each file $W_n$ of size $F$ bits is divided into two sub-files $W_n^1$ and $W_n^2$ of size $F_1$ and $F_2$ bits, respectively, where $F_1=\alpha F$ and $F_2 = (1-\alpha) F$ for some $0 \leq \alpha \leq1$, i.e.,
$$W_n \rightarrow W^1_n \quad \text{and} \quad W^2_n, \quad \forall n \in [N].$$ We further divide $W_n^1$ into $K$ mini-subfiles, denoted by $W^1_{n,1}, W^1_{n,2}, \ldots, W^1_{n,K}$ for all $n \in [N]$, and divide $W_n^2$ into ${K \choose t}$ mini-subfiles, denoted by $W^2_{n,S}$, where $S \subset [K]$ and $|S|=t$ for all $n \in [N]$, and $t \in \{1, 2, \ldots, K_1K_2\}$. Clearly, the size of $W^1_{n,i}$ is $\frac{\alpha F}{K}$ for all $n\in [N], i \in [K],$ where $K=K_1 K_2$ and the size of $W^2_{n,S}$ is $\frac{(1-\alpha) F}{{K \choose t}}$,  $ \forall n \in [N]$, $S \subset [K]$ and $|S|=t$. We use the CFL scheme for the subfiles $W_n^1, n\in [N]$ and the MN scheme for the subfiles $W_n^2, n \in [N]$. Since the CFL scheme works for $N \leq K$, in our scheme, we only consider $\alpha=0$ for $N>K$. 

\subsection{Placement phase}
Let the cache content of mirror $m$ be denoted by $\Lambda_m$ for $m \in [K_1]$ and the cache content of user $k$ be denoted by $Z_k$ for $k \in [K]$. Define a set $S_m = \{(m-1)K_2+1, (m-1)K_2+2, \ldots, mK_2 \}$ which contains all the users connected to mirror $m$. This formula ensures that each mirror $m$ has a unique, non-overlapping set of $K_2$ users. The content placement for both the mirror and the users differ depending on $t < K_{2}$ and $t \geq K_{2}$ as follows.

\begin{itemize}
    \item For $t < K_{2}$,
    \begin{equation}
    \Lambda_m = \{ W^1_{1,k} \oplus W^1_{2,k} \oplus \cdots \oplus W^1_{N,k}  \ | \ \forall k \in S_m \} \quad \text{for } m \in [K_1],
         \label{mirrplace}
    \end{equation}
    and
    \begin{equation}
    Z_k = \{ W^2_{n,S} \ | \ S \subseteq [K],\ |S|=t,\ \forall n \in [N] \} \quad \text{for } k \in S.
         \label{userplace}
    \end{equation}
\end{itemize}

The memory sizes of the mirror and the user caches for this case are given by:

\begin{align}
    M_1 &=\frac{\alpha K_2}{K} 
    \label{memory1_t_less_k2} \\
 M_2 &= \frac{(1-\alpha)}{{K \choose {t}}}\left[ {K-1 \choose{t-1}} \right] N
     \label{memory2_t_less_k2}
\end{align}
\begin{itemize}
    \item For $t \geq K_{2}$,
\begin{align*}
 \Lambda_m = &\{ W^1_{1,k} \oplus W^1_{2,k} \oplus \cdots \oplus W^1_{N,k}  \ | \ \forall k \in S_m\}  \cup \\ &  \{ W^2_{n,S} \ | \ S \subseteq [K],\ |S|=t,\ S_m\subseteq [S] ,  \&    \ \forall n\in [N] \} 
\end{align*}
and
\begin{align*}
Z_k =  \{ W^2_{n,S} \ | \ S \subseteq [K],\ |S|=t,\ k\in S , S_m\not\subseteq [S],  \ \text{and}\ \forall n\in [N] \}
\end{align*}
\end{itemize}

Therefore, for both the cases $t<K_2$ and $t \geq K_2$, we have
\begin{equation}
    M_1 =\frac{\alpha K_2}{K} + \frac{(1-\alpha)}{{K \choose t}}{{K-K_{2} \choose{t-K_2}}} N,
    \label{memory1}
\end{equation}
\begin{equation}
     M_2 = \frac{(1-\alpha)}{{K \choose {t}}}\left[ {K-1 \choose{t-1}} -{K-K_{2} \choose{t-K_2}} \right] N.
     \label{memory2}
\end{equation} 
Considering the value of ${a\choose b}$ is zero if $b$ is negative, for all $1 \leq t \leq K$, the value of ${M_1}$ ranges from $\frac{{K-K_{2} \choose{t-K_2}} N}{{K \choose t}}$ to $\frac{1}{K_1}$, while ${M_2}$ is related to ${M_1}$ via the following equation:
\begin{equation}
M_2 = \frac{\left(M_1  - \frac{1}{K_1}\right) \left( {K_1K_2-1 \choose{t-1}} -{K_1K_2-K_2 \choose{t-K_2}} \right)N}{\left({\frac{{K-K_{2} \choose{t-K_2}} N}{{K \choose {t}}}} - \frac{1}{K_1} \right){K_1K_2 \choose t}}\ \, .
\label{m2equation}
\end{equation}
Clearly, the variation of ${M_2}$ with respect to ${M_1}$ is linear. Note that $M_1$ and $M_2$ are constrained unlike in \cite{karamchandani2016hierarchical, 8481555, wang2019reduce}. However, some other constraints appear in \cite{kong2022centralized}. Further, we have
$$M_1+M_2=\frac{\alpha K_2}{K} + \frac{(1-\alpha)Nt}{K}, $$ 
and the global memory of the system is, 
\begin{align}
\overline{M}&=M_1K_1+M_2K \nonumber\\
&=\alpha +  \frac{(1-\alpha)}{{K \choose {t}}}\left[ \frac{1}{K_2}{K-K_{2} \choose{t-K_2}} + {K-1 \choose{t-1}}
  -{K-K_{2} \choose{t-K_2}} \right]NK .
\label{global memory}
\end{align}
For a fixed $t \in [K]$, the proposed scheme works for $$\overline{M} \in \left[ 1, \frac{1}{{K \choose {t}}}\left( \frac{1}{K_2}{K-K_{2} \choose{t-K_2}} + {K-1 \choose{t-1}} -{K-K_{2} \choose{t-K_2}} \right) NK\right]$$ which corresponds to $\alpha \in [0,1].$

\subsection{Delivery phase}
Let the demand vector be $\overline{d}=(d_1,d_2, \ldots, d_K)$, and each file is demanded by at least one user. Consider a set $\mathcal{B} \subseteq [K]$ (called the base set) such that $|\mathcal{B}|=N$ and $\{d_i \ | \ i \in \mathcal{B}\} = [N]$.
The transmissions are as follows:
\\
\noindent \textbf{Server to mirrors:}
The following transmissions are made from the server to the mirrors.
\begin{enumerate}[label=(\subscript{\textbf{SM}}{{\arabic*}})]
   \item For each $i\in [K]$, transmit mini-subfiles $W^1_{n,i}$, for $n\in [N]$ and $n \neq d_i$.
   \item For each $i \in [K] \backslash \mathcal{B}$, transmit $W^1_{d_i, i} \oplus W^1_{d_{i'},i'}$, where $i' \in \mathcal{B}$ and $d_i=d_{i'}$.
   \item For each $S \subseteq [K]$ such that $|S|=t+1$, transmit $\bigoplus_{s\in S} W^2_{d_s, S\backslash \{s\}}$. 
\end{enumerate}

\noindent \textbf{Mirror to users:}
For $m \in [K_1]$, consider $D_{m}=\{ d_k \ | \ k \in S_m \}$ and $|D_{m}|=t_m$.  Clearly, $1 \leq t_m \leq K_2$. The following transmissions are made from the $m^{\text{th}}$ mirror to the users attached to it, i.e., all users in set $S_m$, where $m \in [K_1]$.
\begin{enumerate}[label=(\subscript{\textbf{MU}}{{\arabic*}})]
\item For each $n \in D_{m} $, transmit $W^1_{n,i}$ for all $i \in [K]$.
\item For each $S \subseteq [K]$ such that $|S|=t+1$ and $S\cap S_m \neq \emptyset$, transmit 
$\bigoplus_{s \in S} W^2_{d_s, S\backslash \{s\}}.$
\item For each $n \in D_{m} $, transmit $W^2_{n, S}$, such that $S \subseteq [K]$, $|S|=t$, $S_m \subseteq [S]$.
\end{enumerate}

\begin{note}
For $N>K$, the steps $(\textbf{SM}_1)$, $(\textbf{SM}_2)$ and $(\textbf{MU}_1)$ are not needed as for this case we only considering $\alpha=0$ and there will be no subfiles named $W_n^1, n \in [N]$. Also, there will be no need for the base set $\mathcal{B}$.
\end{note}

\subsection{Proof of correctness}
\subsubsection {Mirror gets from the server what it transmits to users:}

In this subsection, we prove that after receiving the transmissions from the server, each mirror will get all the files it transmits to the users in the delivery phase.

For $m \in [K_1]$, consider $m^{\text{th}}$ mirror with the set $D_m=\{d_k \ | \ k \in S_m\}$. Let the set $D^{(n)}$ contain all the users with demand $W_n$ for all $n \in [N]$, i.e., 
$$D^{(n)}=\{k \in [K] \ | \ d_k=n\}.$$
We will prove this step-wise:

\noindent \textbf{For step $(\textbf{MU}_1)$:} Let $\lambda \in S_m$, then after receiving the transmissions from the server, mirror $m$ should have all the mini-subfiles of $W^1_{d_{\lambda}}$, i.e., $W^1_{d_{\lambda},1}, W^1_{d_{\lambda},2}, \ldots, W^1_{d_{\lambda},K}$. From step $(\textbf{SM}_1)$, the mirror receives $W^1_{n,\lambda}$ for all $n \in [N]$ and $n \neq d_{\lambda}$. Hence, using mini-subfile $W^1_{1,\lambda} \oplus W^1_{2,\lambda} \oplus \cdots \oplus W^1_{N,\lambda}$ from the cache of mirror $m$, we get $W^1_{d_{\lambda},\lambda}$. Also, from step $(\textbf{SM}_1)$, we get $W^1_{d_{\lambda}, j}$ for all $j \in [K]$ such that $d_j \neq d_{\lambda}$. Now, the following mini-subfiles are left to be retrieved.
$$W^1_{d_{\lambda}, j}, \ \forall j \in D^{(d_{\lambda})} \backslash \{\lambda\}.$$
There are following two cases:
\begin{enumerate}
\item If $\lambda \in \mathcal{B}$, then $D^{(d_{\lambda})} \backslash \{\lambda\} \subseteq  [K]\backslash \mathcal{B}$. Hence from step $(\textbf{SM}_2)$, the mirror receives 
$$W^1_{d_{j}, j} \oplus W^1_{d_{j}, \lambda} = W^1_{d_{\lambda}, j} \oplus W^1_{d_{\lambda}, \lambda},$$
 for all $j \in D^{(d_{\lambda})} \backslash \{\lambda\}$. Since the mirror already has $W^1_{d_{\lambda}, \lambda}$, it will get $W^1_{d_{\lambda}, j}$. 
\item  If $\lambda \not\in \mathcal{B}$, then there exist $\lambda' \in D^{(d_{\lambda})}$ such that $\lambda' \in \mathcal{B}$. Since $\lambda \in [K] \backslash \mathcal{B}$, from step $(\textbf{SM}_2)$, the mirror receives 
$W^1_{d_{\lambda}, \lambda} \oplus W^1_{d_{\lambda}, \lambda'}$. The mirror will get $W^1_{d_{\lambda}, \lambda'}$ as it already has $W^1_{d_{\lambda}, \lambda}$. Again from step $(\textbf{SM}_2)$, the mirror receives 
$$W^1_{d_{j}, j} \oplus W^1_{d_{j}, \lambda'} = W^1_{d_{\lambda}, j} \oplus W^1_{d_{\lambda}, \lambda'},$$
 $\forall j \in D^{(d_{\lambda})} \backslash \{\lambda, \lambda'\}$. Since the mirror already has $W^1_{d_{\lambda}, \lambda'}$, it will get $W^1_{d_{\lambda}, j}$. 
\end{enumerate}

\noindent \textbf{For step $(\textbf{MU}_2)$:} In step $(\textbf{SM}_3)$, the server transmits $\bigoplus_{s\in S} W^2_{d_s, S\backslash \{s\}}$ for all $S \subseteq [K]$ such that $|S|=t+1$. Therefore, mirror $m$ will get $\bigoplus_{s\in S} W^2_{d_s, S\backslash \{s\}}$  for all
$S \subseteq [K]$ such that $|S|=t+1$ and $S \cap S_m \neq \emptyset$, transmitted in step $(\textbf{MU}_2)$. 

\noindent \textbf{For step $(\textbf{MU}_3)$:} Since cache of mirror $m$ contains, $W^2_{n,S_m}$ for all $n \in [N]$, it has all the mini-subfiles transmitted to users in step  $(\textbf{MU}_3)$.

\subsubsection{All users get their desired files:}

In this subsection, we prove that after receiving the transmissions from the mirrors, each user will get the demanded file in the delivery phase.

Consider a user $\lambda \in S_m$ for some $m \in [K_1]$. The demand of user $\lambda$ is $W_{d_{\lambda}}$. In step $(\textbf{MU}_1)$, all the mini-subfiles of $W^1_{d_{\lambda}}$ are transmitted from mirror $m$ to all users in $S_m$. Therefore, we need to show that from step $(\textbf{MU}_2)$, step $(\textbf{MU}_3)$ and the cache content of user $\lambda$, $Z_{\lambda}$, user $\lambda$ will get the sub-file $W^2_{d_{\lambda}}$. Since from step $(\textbf{MU}_3)$, user $\lambda$ will get $W^2_{d_{\lambda}, S_m}$, the user has the following mini-subfiles of $W^2_{d_{\lambda}}$,
$$W^2_{d_{\lambda}, S}, \ \text{where} \ S\subseteq [K], |S|=t, \ \text{and} \ \lambda \in S.$$
Now take a mini-subfile $W^2_{d_{\lambda}, S'}$, where $ S' \subseteq [K], |S'|=K_2, \ \text{and} \ \lambda \not\in S'$. Since $\lambda \in S_m$, the set $S' \cup \{\lambda\}$ will satisfy the condition of step $(\textbf{MU}_2)$, and the user will receive the transmission 
$$\bigoplus_{s\in S' \cup\{ \lambda \}} W^2_{d_s, S' \cup \{ \lambda \}\backslash \{s\}}.$$
From this transmission, user $\lambda$ will get $W^2_{d_{\lambda}, S'}$ as it already has all the other mini-subfiles $W^2_{d_{s}, S' \cup \{ \lambda \}\backslash \{s\}}$ for $s \neq \lambda$. 
\subsection{Rate}
\noindent \textbf{Server to mirrors:}
As per the delivery phase, the following rate is calculated step-wise.
\begin{enumerate}
\item The total number of mini-subfiles of size $F_1$ transmitted in step $(\textbf{SM}_1)$ and $(\textbf{SM}_2)$ together is $K(N-1) + K-N = N(K-1)$.
\item The total number of mini-subfiles of size $F_2$ transmitted in step $(\textbf{SM}_3)$ is ${K \choose{t+1}}$.
\end{enumerate}
Therefore, we have the rate:
\begin{align}
R_1 &=  \frac{\alpha}{K} [N(K-1)]+  \frac{(1-\alpha)}{{K \choose{t}}}\left[ {K \choose{t+1}} \right] \nonumber \\
&= \alpha \frac{N(K-1)}{K}+ (1-\alpha)  \frac{K-t}{t+1}.
\label{rate1}
\end{align}

\noindent \textbf{Mirror to users:}
As per the delivery phase, the following rate is calculated step-wise for $m^{\text{th}}$ mirror to the users connected to it, where $m \in [K_1]$.
\begin{enumerate}
\item The total number of mini-subfiles of size $F_1$ transmitted in step $(\textbf{MU}_1)$ is $Kt_m$, where $t_m= | D_m|$.
\item The total number of mini-subfiles of size $F_2$ transmitted in step  $(\textbf{MU}_2)$ is ${K \choose{t+1}} - {K-K_2 \choose{t+1}}$.
\item The total number of mini-subfiles of size $F_2$  transmitted in step $(\textbf{MU}_3)$ is ${K-{K_2}} \choose {t-K_{2}}$.
\end{enumerate}
 Therefore, the total number of mini-subfiles of size $F_1$ is $Kt_m$, and the total number of mini-subfiles of size $F_2$ is 
${K \choose{t+1}}$-${K-{K_2} \choose{t+1}}$ + ${K-{K_2}} \choose {t-K_{2}}$. Hence, the rate from $m^{\text{th}}$ mirror to the users connected to it is given as follows: 
\begin{align}
R^{(m)}_2 &=\frac{\alpha Kt_m}{K}+ \frac{(1-\alpha)\left( {K \choose{t+1}} - {K-K_2 \choose{t+1}} + {K-K_2 \choose t-K_2} t_m \right)}{{K \choose{t}}}  \nonumber \\
&= \alpha t_m+ (1-\alpha) \left[ \frac{K-t}{t+1} - \frac{{K-K_2 \choose{t+1}} -{K-K_2 \choose t-K_2} t_m}{{K \choose{t}}} \right]
\label{rate2}\ \, ,
\end{align}
where $1 \leq t_m \leq K_2$.
If $N=K$, then for all $m \in [K_1]$, we have $t_m=K_2$  and
$$ R^{(m)}_2=\alpha K_2+ (1-\alpha) \left[ \frac{K-t}{t+1} - \frac{{K-K_2 \choose{t+1}} -{K-K_2 \choose t-K_2} K_2}{{K \choose{t}}} \right].$$ 
Let \( R_2 \) denote the worst-case rate from a mirror to users when each user attached to one mirror demands at least a distinct file., i.e., for at least one $m\in [K_1]$, we have $t_m=K_2$. Therefore,
$$ R_2=\alpha K_2+ (1-\alpha) \left[ \frac{K-t}{t+1} - \frac{{K-K_2 \choose{t+1}} -{K-K_2 \choose t-K_2} K_2}{{K \choose{t}}} \right].$$
Clearly,
$R_2^{(m)} \leq R_2$ for all $m \in [K_1]$.
Now the composite rate, $\overline{R}$ is defined as
$
\overline{R} = R_{1}+K_{1}R_{2}.
$
which represents the overall bandwidth of the system.
In Figure \ref{fig:MR}, the plot between the global memory and composite rate has been shown for the proposed scheme with the following observations:
\begin{itemize}
\item Corresponding to each $t \in [K]$, there is a line which is obtained by moving $\alpha$ from $0$ to $1$. Since the proposed scheme is a combination of the CFL scheme and the MN scheme, the end point of a line where $\alpha=1$ corresponds to only the CFL scheme and other end point where $\alpha=0$ corresponds to only the MN scheme.
\item For all $t < K_2$, mirrors' caches only contain the coded placements, and at the end point $(\alpha=0)$, $M_1=0$. Whereas for $K_2 \leq t \leq K$, mirrors' caches contain coded and uncoded placements, which can be seen from the given placement phase.
\item For a fixed global memory $\overline{M}=m$, we get multiple global memory-composite rate points $(\overline{M}, \overline{R})$. Each point corresponds to a different pair of $(t, \alpha)$, which means a different pair of cache memories $(M_1, M_2)$. If the distribution of the global memory is not predefined in terms of $M_1$ and $M_2$, choosing the lowest point in the plot will get the minimum composite rate.
\item All the end points $(\alpha=0)$ of the lines for $K_2 < t <K$ correspond to the first scheme given in \cite{kong2022centralized} which is also described in Subsection \ref{alpha0}.
\item Since the CFL scheme gives a better rate than the MN scheme only when $N\leq K$, so for the case $N>K$, we consider only the memory points for $\alpha=0$ which corresponds to MN scheme only, for all $t \in [K]$ which are shown in Figure \ref{fig:MR2}.
\end{itemize}
In environments where mirrors are positioned in close proximity, preventing signal interference becomes crucial. This is achieved by ensuring that all mirrors transmit using orthogonal carrier frequencies, thereby underlining the utility of the composite rate, which represents the total bandwidth utilized. 
When all mirrors are situated at considerable distances from each other, they can utilize the same carrier frequency. For these cases, the metric $R_1$ + $R_2$ becomes a more relevant measure for analyzing total bandwidth requirements. Additionally, we also plot the tradeoff between \( R_1 \) and \( R_2 \) separately with \( M_1 \) and \( M_2 \), i.e., \( R_1 \) with \( M_1 \) and \( M_2 \), and then \( R_2 \) with \( M_1 \) and \( M_2 \). This is because in scenarios where concurrent transmission between the server to mirror and mirror to users is allowed, the metric \( \max(R_1, R_2) \) becomes a more relevant measure to analyze the bandwidth requirement. 
We plot $R_1$ + $R_2$ against global memory for our proposed schemes and existing schemes, providing a comprehensive comparison. 
Additionally, we provide a detailed visualization by plotting $R_1$ + $R_2$ separately against $M_1$ and $M_2$. This offers a nuanced perspective on performance analysis with different pairs of $M_1$ and $M_2$, which may result in the same global memory.


Furthermore, we illustrate the trade-off between $R_1$ and $R_2$ separately with $M_1$ and $M_2$ within our proposed scheme, comparing it with existing methodologies. Section \ref{compare} presents a detailed comparison with each scheme.
\subsection{Memory Sharing} \label{memory sharing}

\begin{thm} If three memory rate points \( \mathbf{A}(M_{1}^{A}, M_{2}^{A}, R_{1}^{A}), \mathbf{B}(M_{1}^{B}, M_{2}^{B}, R_{1}^{B}), \) and \( \mathbf{C}(M_{1}^{C}, M_{2}^{C}, R_{1}^{C}) \) are achievable by a scheme, then any point \( \mathbf{P}(M_{1}, M_{2}, R_{1}) \) in the convex hull of \( \mathbf{A}, \mathbf{B}, \) and \( \mathbf{C} \) (which forms a triangle with vectors \( \mathbf{A}, \mathbf{B}, \) and \( \mathbf{C} \)) is achievable.
\end{thm}

\textbf{Proof:}
Let \( (M_{1}, M_{2}) \) lie in the convex hull of \( \mathbf{A}(M_{1}^{A}, M_{2}^{A}, R_{1}^{A}) \), \( \mathbf{B}(M_{1}^{B}, M_{2}^{B}, R_{1}^{B}) \), and \( \mathbf{C}(M_{1}^{C}, M_{2}^{C}, R_{1}^{C}) \) as shown in Fig. \ref{mems}. 
From our scheme, we consider line \( \mathbf{AB} \) and line \( \mathbf{AC} \) as two different lines for consecutive \( t \) values, say \( t_1 \) and \( t_1 + 1 \). Point \( \mathbf{A} \) is where \( \alpha = 1 \), and all lines for different values of \( t \) from our scheme originate from this point. Points \( \mathbf{B} \) and \( \mathbf{C} \) correspond to \( \alpha = 0 \) each for \( t \) values \( t_1 \) and \( t_1 + 1 \), respectively. 
At point \( \mathbf{A} \), i.e., at \( \alpha = 1 \), we have \( (M_{1}, M_{2}, R_{1}) \) as \( \left( \frac{K_{2}}{K}, 0, \frac{N(K-1)}{K} \right) \). Then we have \( 0 \leq \xi \leq 1 \) and \( 0 \leq \eta \leq 1 \) such that
\[
(M_{1}, M_{2}) = \xi (M_{1}^{A}, M_{2}^{A}) + \eta (M_{1}^{B}, M_{2}^{B}) + (1 - \xi - \eta) (M_{1}^{C}, M_{2}^{C})
\]

\begin{figure}[H]
    \begin{center}
        \captionsetup{justification=centering}
        \resizebox{\textwidth}{!}{\includegraphics{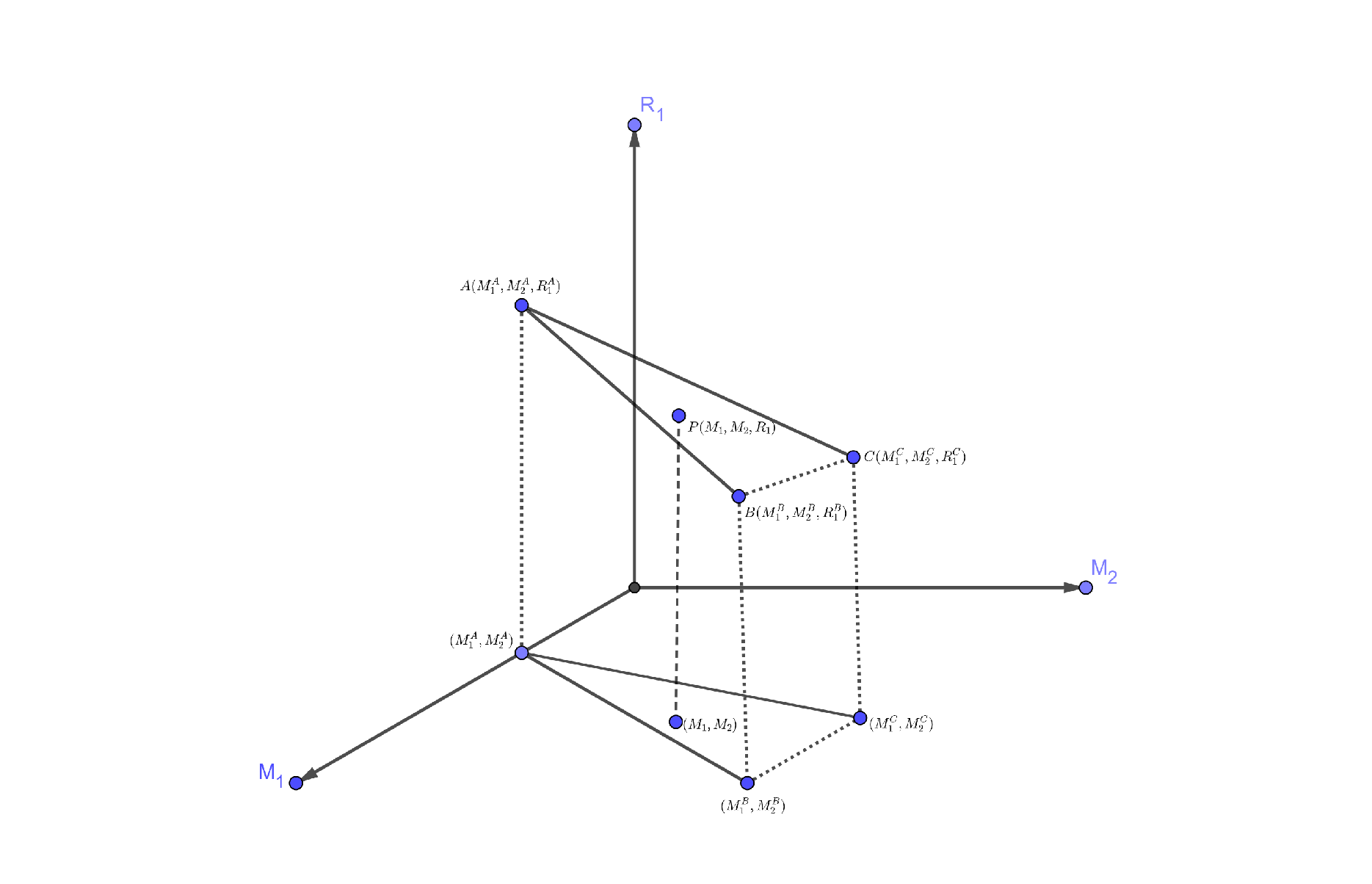}}
        \caption{Convex hull for memory sharing.}
        \label{mems}
    \end{center}
\end{figure}

Divide each file \( W_{n} \), \( n \in [N] \) of size \( F \) bits into three parts as follows:
\[
F = \xi F + \eta F + (1 - \xi - \eta)F = F_{A} + F_{B} + F_{C}
\]
where 
\[
F_{A} = \xi F, \quad F_{B} = \eta F, \quad F_{C} = (1 - \xi - \eta)F
\]
Each file \( W_{n}, n \in [N] \) is divided into three parts, \( W_{n}^{1}, W_{n}^{2}, \) and \( W_{n}^{3} \) of sizes \( F_{1}, F_{2}, \) and \( F_{3} \) respectively.

Now, by using the scheme given for \( \mathbf{A}(M_{1}^{A}, M_{2}^{A}, R_{1}^{A}), \mathbf{B}(M_{1}^{B}, M_{2}^{B}, R_{1}^{B}), \) and \( \mathbf{C}(M_{1}^{C}, M_{2}^{C}, R_{1}^{C}) \) for the subfiles \( W_{n}^{A}, W_{n}^{B}, \) and \( W_{n}^{C} \) respectively, we get the rate at \( (M_{1}, M_{2}) \) as
\[
\begin{aligned}
    RF &= R_{1}^{A} F_{A} + R_{1}^{B} F_{B} + R_{1}^{C} F_{C} \\
       &= R_{1}^{A} (\xi F) + R_{1}^{B} (\eta F) + R_{1}^{C} (1 - \xi - \eta)F \\
       &= \xi R_{1}^{A} + \eta R_{1}^{B} + R_{1}^{C} (1 - \xi - \eta)
\end{aligned}
\]
\begin{note}
Similarly, the same convex hull methodology can be applied to determine the rate of the second layer, \( R_2 \), as was employed for \( R_1 \). Both \( R_1 \) and \( R_2 \) are linear functions of \( \alpha \), \( t \), and the system parameters \( K_1, K_2, N \).
\end{note}
\subsection{Coding Delay} \label{coding delay}

In \cite{wang2019reduce}, the authors designated $T_{d}$ as the duration of the delivery phase. They defined coding delay as the normalization of $T_{d}$ by the file size, offering a standardized metric for delay assessment as:
\begin{equation}
    T  \overset{\Delta}{=} sup \frac{T_{d}}{F}
\end{equation}
The authors posited that if in a ($\emph{$K_{1}$},\emph{$K_{2}$};\emph {$M_{1}$},\emph{$M_{2}$};\emph{N}$)  hierarchical coded caching problem, both the server and all mirrors transmit symbols concurrently across all transmission slots, then the resultant coding delay would be:
\begin{equation}
T = \max\{R_1, R_2\} \label{cd}
\end{equation}
On the contrary, if there is a relay that begins transmission only after the server completes its transmission, then:
\begin{equation}
T = R_1 + R_2
\end{equation}
In our scheme, with concurrent transmission both from the server to mirrors and from mirrors to the server, we facilitate the transfer of cached content from mirrors relevant to the demands of the users attached to them. Consequently, this capability reduces the rate of the second layer, \( R_2 \). We have discussed this possibility in Ex.\ref{ex1} by illustrating the variation of \( R_2 \)  when concurrent transmission is permitted compared to when it's not. 
To our understanding, concurrent transmission has not been addressed in any existing schemes. Hence, we solely compared the coding delay of our scheme with that of the scheme proposed in \cite{wang2019reduce}.

\subsection{For $\alpha=0$:} \label{alpha0}

 The proposed scheme contains the KWC scheme given in \cite{kong2022centralized}, which is based on the MN scheme. For $\alpha =0$, the parameters of the proposed scheme are given as follows:

\begin{align*}
\frac{M_1}{N}& = \frac{1}{{K \choose t}}{{K-K_{2} \choose{t-K_2}}}, \\
  \frac{M_2}{N} &= \frac{1}{{K \choose {t}}}\left[ {K-1 \choose{t-1}} -{K-K_{2} \choose{t-K_2}} \right], \\
R_1&= \frac{K-t}{t+1},\\
R_2& = \frac{K-t}{t+1} - \frac{{K-K_2 \choose{t+1}} -{K-K_2 \choose t-K_2} K_2}{{K \choose{t}}}.
\end{align*}
These parameters are the same as the parameters of the first scheme given in \cite{kong2022centralized} for $K_2 < t < K$.

\begin{figure*}[!t]
	\centering
		\includegraphics[width= 0.8\textwidth]{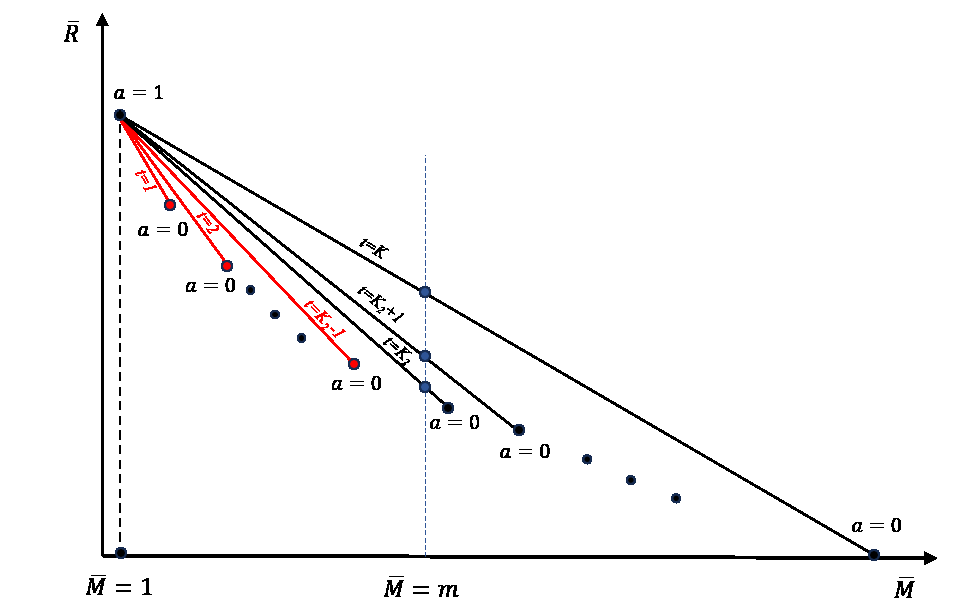}
		\caption{Global memory and composite rate for $N \leq K$}
		\label{fig:MR}	
\end{figure*} 

\begin{figure*}[!t]
	\centering
		\includegraphics[width= 0.8\textwidth]{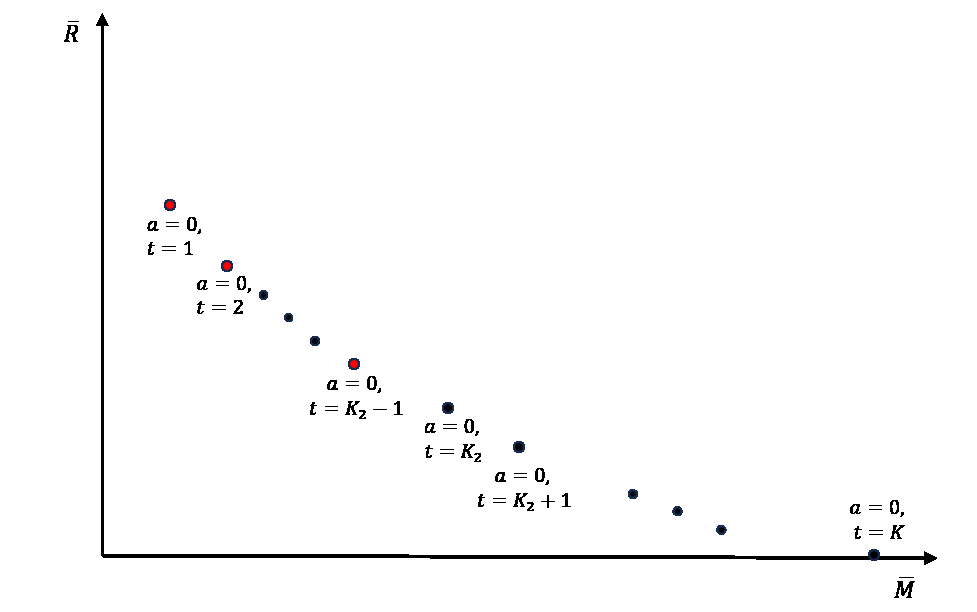}
		\caption{Global memory and composite rate for $N > K$}
		\label{fig:MR2}	
\end{figure*} 

\begin{figure*}[!t]
	\centering
		\includegraphics[width= 0.8\textwidth]{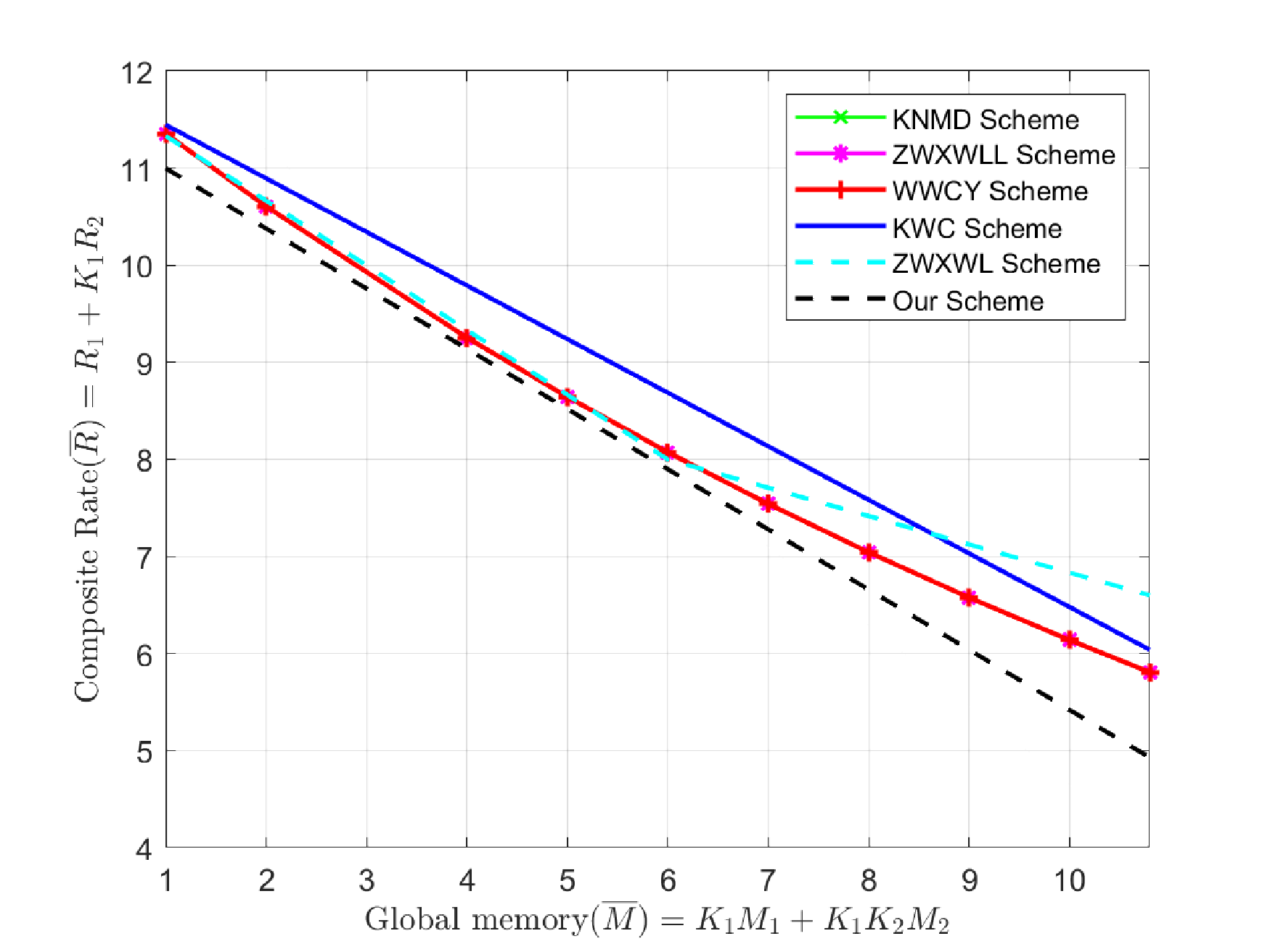}
		\caption{Performance comparison for a ($3,2$;\emph {$M_{1}$},\emph{$M_{2}$};6) hierarchical caching system in Example \ref{ex1}}
		\label{fig:CRvsGM}	
\end{figure*} 

The following example illustrates the placement and delivery phase of the proposed scheme for the case when $N=K$, whereas Example \ref{ex2} explains it for the case $N<K$.
\begin{example} \label{ex1}
Let $K_1=3, K_2=2, N=6,$ $\alpha = \frac{1}{2}$. For $t=2$, we have,
$$M_1=\frac{\alpha K_2}{K}  + \frac{(1-\alpha)N}{{K \choose {t}}}=\frac{11}{30}, $$
$$ \\M_2= \frac{(1-\alpha)}{{K \choose {t}}}\left( {K-1 \choose{t-1}} -1 \right) N = \frac{4}{5},$$
$M_1+M_2=\frac{7}{6}$ and the global memory is $\overline{M}=5.9$.
Divide each file $W_n$ into two sub-files $W^1_n$ and $W^2_n$, for all $n \in [6]$. Further, divide $W^1_n$ into $6$ mini-subfiles, say, $W^1_{n,1}, W^1_{n,2}, \ldots, W^1_{n,6}$, and divide $W^2_n$ into $15$ mini-subfiles, say, $W^2_{n,S},  S \in [6], |S|=2$. Consider the sets \\
$$S_1=\{ 1,2 \}, \quad S_2=\{3,4\}, \quad S_3=\{5,6\}.$$

\begin{figure*}[!t]
	\centering
		\includegraphics[width= 0.8\textwidth]{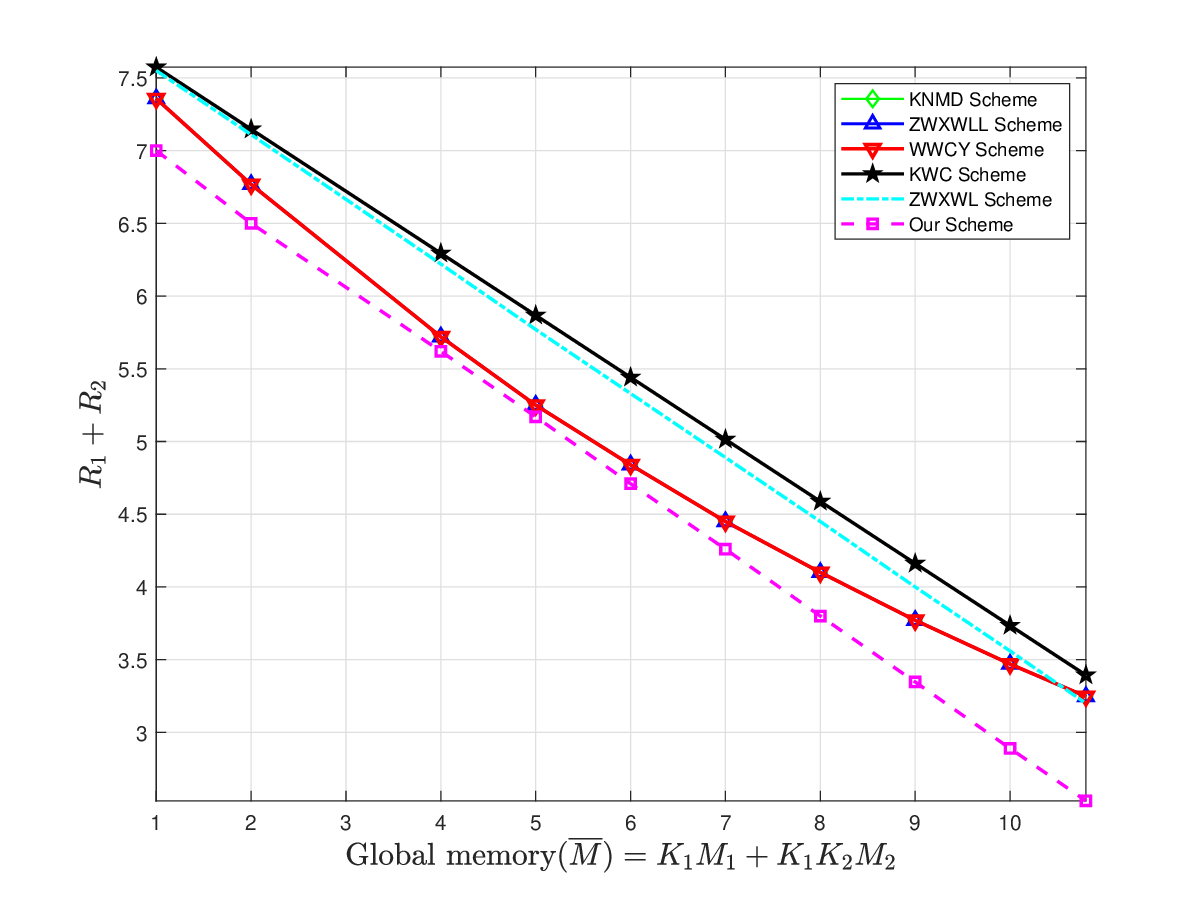} \vspace{-0.5cm}
		\caption{$R_1$+$R_2$ trade-off for ($3,2$;\emph {$M_{1}$},\emph{$M_{2}$};6) hierarchical caching system in Example \ref{ex1}}.
		\label{R_1+R_2}	
\end{figure*}

\noindent \textbf{Placement phase:} The cache contents of the mirrors are as follows,
\begin{align*}
\Lambda_1 &=\{ W^1_{1,k}  \oplus \cdots \oplus W^1_{6,k} , \forall k \in S_1\} \cup  \{W^2_{n,\{1,2\}} ,\forall n \in [6]\}\\
\Lambda_2 &= \{ W^1_{1,k}  \oplus \cdots \oplus W^1_{6,k}, \forall k \in S_2\} \cup  \{W^2_{n,\{3,4\}}, \forall n \in [6]\}\\
\Lambda_3 &=\{ W^1_{1,k} \oplus \cdots \oplus W^1_{6,k} ,\forall k \in S_3\} \cup  \{W^2_{n,\{5,6\}} , \forall n \in [6]\}.
\end{align*}
The cache contents of the users are as follows,
\begin{align*}
Z_1 &=\{ W^2_{n,\{1,3\}}, W^2_{n,\{1,4\}}, W^2_{n,\{1,5\}}, W^2_{n,\{1,6\}}, \forall n\in [6] \}\\
Z_2 &= \{ W^2_{n,\{2,3\}}, W^2_{n,\{2,4\}}, W^2_{n,\{2,5\}}, W^2_{n,\{2,6\}} ,\forall n\in [6] \}\\
Z_3 &=\{ W^2_{n,\{1,3\}}, W^2_{n,\{2,3\}}, W^2_{n,\{3,5\}}, W^2_{n,\{3,6\}} ,\forall n\in [6] \}\\
Z_4 &= \{ W^2_{n,\{1,4\}}, W^2_{n,\{2,4\}}, W^2_{n,\{4,5\}}, W^2_{n,\{4,6\}}, \forall n\in [6] \}\\
Z_5 &=\{ W^2_{n,\{1,5\}}, W^2_{n,\{2,5\}}, W^2_{n,\{3,5\}}, W^2_{n,\{4,5\}} , \forall n\in [6] \}\\
Z_6 &= \{ W^2_{n,\{1,6\}}, W^2_{n,\{2,6\}}, W^2_{n,\{3,6\}}, W^2_{n,\{4,6\}}, \forall n\in [6] \}.\\
\end{align*} 

\noindent \textbf{Delivery phase:} Let the demand vector be $(1,2,3,4,5,6)$,  i.e.,\ all users demand distinct files, and $\mathcal{B}=\{1,2,3,4,5,6\}$. The delivery takes place in two phases. The server transmits to all the mirrors, followed by the transmission from the mirrors to each of the attached users respectively. 

\noindent  \textbf{Server to mirrors:}
\begin{itemize}
\item The server transmits the following uncoded messages
$W^1_{i,j}$, 
where, $1 \leq i,j \leq 6$, and $i \neq j$.
\item The server also transmits the coded messages:
$$\bigoplus_{s\in S} W^2_{d_s, S\backslash \{s\}},$$ for all $S \subseteq [6]$ such that $|S|=3$.
\end{itemize}
Therefore, the rate is 
$$R_1=\frac{30}{2 \times 6}+ \frac{20}{2\times 15}= \frac{19}{6}= 3.167.$$

\noindent \textbf{Mirrors to users:} The $1^{\text{st}}$ mirror will transmit the following uncoded and coded messages:		
\begin{itemize}

\item Uncoded messages:\\
$W^1_{1,1},  W^1_{1,2}, W^1_{1,3}, W^1_{1,4}, W^1_{1,5}, W^1_{1,6}, 
W^1_{2,1}, W^1_{2,2}, W^1_{2,3},  W^1_{2,4},  W^1_{2,5},  W^1_{2,6}$.
\\
\item Coded messages: 
$$\bigoplus_{s\in S} W^2_{d_s, S\backslash \{s\}},$$ 
for all $S \in \mathcal{S} \backslash \{ \{3,4,5\}, \{3,4,6\}, \{3,5,6\}, \{4,5,6\}\}$, where $\mathcal{S}=\{ T \subseteq [6] \ | \ |T|=3\}$.
\\
\item The mirror also transmits $W^2_{1,\{1,2\}}$ and $W^2_{2,\{1,2\}}$, which were placed in the mirror during placement phase.
\end{itemize}
Therefore, the rate is 
$$R_2=\frac{12}{2 \times 6}+ \frac{18}{2\times 15}= \frac{8}{5}=1.6.$$

Similarly, the rate from the  $2^{\text{nd}}$ mirror to the users in the set $S_2$ and the rate from the $3^{\text{rd}}$ mirror to the users in the set $S_3$ are equal to $R_2=1.6$. Therefore, the proposed scheme achieves the composite rate $\overline{R}=7.967$ for the global memory is $\overline{M}=5.9$. For the same value of $N, K$, and $\overline{M}$, the KWC scheme gives the composite rate $\overline{R}=8.743$.
\end{example}

As discussed in Section \ref{coding delay}, If we incorporate transmission from the mirrors to the users simultaneously with the server transmitting to the mirrors, we can send $W^2_{1,\{1,2\}}$ and $W^2_{2,\{1,2\}}$, stored in the mirror during the placement phase, which reduces $R_2$ to $1.533$. For the same memory points ($M_1$,$M_2$), WWCY scheme \cite{wang2019reduce} achieves $R_1=3.264$ and $R_2=1.62$ through permitted concurrent transmission. 
In this case, the coding delay, as defined in Equation \ref{cd}, is $T=3.264$ for the WWCY scheme and $T=3.167$ for our scheme.

For $N=6, K_1=3, K_2=2$ and $t=2$, Figure \ref{fig:CRvsGM} shows the comparison of the proposed scheme with the other existing schemes using the parameters global memory and composite rate, which is also discussed in the next section in detail. In Figure \ref{fig:CRvsGM}, the line corresponding to the proposed scheme is obtained by varying $\alpha$ from $0$ to $1$. Additionally, we analyze the trade-off for $R_1$+$R_2$ with global memory in Fig. \ref{R_1+R_2}. Our scheme yields the minimum $R_1$+$R_2$, which is a more relevant metric to analyze scenarios where all the mirrors utilize identical bandwidth. Furthermore, for the same parameters $N, K_1, K_2$ and $\overline{M}=7.2$ (corresponding to $t=2$ and $\alpha=\frac{18}{49}$), the rates of the proposed scheme are compared with the existing schemes in Table \ref{tab1}. We kept the value of global memory the same for all the existing schemes and computed the value of $R_1$, $R_2$, and the composite rate $\overline{R}$. It is clear from the table that the proposed scheme achieves the lowest composite rate ($\overline{R}=7.162$) for the global memory $\overline{M}=7.2$. The value of $R_2$ is minimum for the proposed scheme, and the value of $R_1$ is also minimum for the proposed scheme except for the ZWXWL scheme. The computation of the rates of the existing schemes is discussed in details in the next section.

\begin{table*}[ht]
	\centering
	\caption{Comparison of schemes with $K_1=3$, $K_2=2$, $N=6$, and $\overline{M}=7.2 \ (t=2, \alpha=\frac{18}{49})$  for Example \ref{ex1}} 
	\begin{tabular}{| c| c | c |c|c|c| }
		\hline
		\rule{0pt}{2.5ex}
		\makecell{} & \makecell{KNMD Scheme\\\cite{karamchandani2016hierarchical}}  & \makecell{ZWXWL Scheme\\\cite{7763243}} & \makecell{ZWXWLL Scheme\\\cite{8481555}} &\makecell{WWCY Scheme\\\cite{wang2019reduce}}  & \makecell{Proposed scheme \\ (For same $\overline{M}$)}  \\ [2pt]
		\hline
  
		\rule{0pt}{2.5ex}
		 $M_{1}$ &  $0.3755$&  $1.8$  &  $0.3755$ &  $0.3755$     & $0.3755$ 
		\\
		\hline
		\rule{0pt}{2.5ex}
		$M_{2}$ & $1.0122$& $0.3$ & $1.0122$ & $1.0122$ &  $1.0122$  \\
		\hline

        \rule{0pt}{2.5ex}
		 $\overline{M}$ & $7.2$& $7.2$  & $7.2$ & $7.2$ &   $7.2$
		\\
		\hline
		\rule{0pt}{2.5ex}
		$R_{1}$ &  $2.8759$&  $2.1$ & $2.857$ & $2.8765$   &  $2.68$  \\
		\hline
		
		\rule{0pt}{3.5ex}
       $R_{2}$ &$1.5323$ &$1.85$&$1.522$ & $1.5323$   & $1.49$
		\\
		\hline

       \rule{0pt}{3.5ex}
       $R_{1}+K_{1}R_{2}$ & $7.4728$& $7.65$ & $7.423$& $7.4734$   &  $7.162$
		\\
		\hline
	\end{tabular}
	
	\label{tab1}
\end{table*}

The following example explains the placement and delivery phase of the case $N<K$.
\begin{example} \label{ex2}
Let $K_1=3, K_2=2, N=3,$ $\alpha = \frac{1}{2}$. For $t=2$, we have, 
$$M_1=\frac{4}{15},  \quad  M_2= \frac{2}{5},$$
$M_1+M_2=\frac{2}{3}$ and $\overline{M}=\frac{16}{5}$.
Divide each file $W_n$ into two sub-files $W^1_n$ and $W^2_n$, for all $n \in [3]$. Further, divide $W^1_n$ into $6$ mini-subfiles, say, $W^1_{n,1}, W^1_{n,2}, \ldots, W^1_{n,6}$, and divide $W^2_n$ into $15$ mini-subfiles, say, $W^2_{n,S},  S \in [6], |S|=2$. Consider the sets 
$$S_1=\{ 1,2 \}, \quad S_2=\{3,4\}, \quad S_3=\{5,6\}.$$
 \noindent \textbf{Placement phase:} The cache contents of the mirrors are as follows,
\begin{align*}
\Lambda_1 &= \{ W^1_{1,k} \oplus W^1_{2,k} \oplus W^1_{3,k}, \forall k \in S_1\} \cup  \{W^2_{n,\{1,2\}} ,\forall  n \in [3]\}\\
\Lambda_2 &=  \{ W^1_{1,k} \oplus W^1_{2,k} \oplus W^1_{3,k},\forall k \in S_2\} \cup  \{W^2_{n,\{3,4\}},\forall  n \in [3] \} \\
 \Lambda_3 &=  \{ W^1_{1,k} \oplus W^1_{2,k} \oplus W^1_{3,k}, \forall k \in S_3\} \cup  \{W^2_{n,\{5,6\}},\forall  n \in [3] \}.
\end{align*}
The cache contents of the users are as follows,
\begin{align*}
Z_1 &=\{ W^2_{n,\{1,3\}}, W^2_{n,\{1,4\}}, W^2_{n,\{1,5\}}, W^2_{n,\{1,6\}} ,\forall n\in [3] \}\\
Z_2 &=\{ W^2_{n,\{2,3\}}, W^2_{n,\{2,4\}}, W^2_{n,\{2,5\}}, W^2_{n,\{2,6\}} ,\forall n\in [3] \}\\
Z_3 &= \{ W^2_{n,\{1,3\}}, W^2_{n,\{2,3\}}, W^2_{n,\{3,5\}}, W^2_{n,\{3,6\}} ,\forall n\in [3] \}\\
Z_4 &= \{ W^2_{n,\{1,4\}}, W^2_{n,\{2,4\}}, W^2_{n,\{4,5\}}, W^2_{n,\{4,6\}} ,\forall n\in [3] \}\\
Z_5 &= \{ W^2_{n,\{1,5\}}, W^2_{n,\{2,5\}}, W^2_{n,\{3,5\}}, W^2_{n,\{4,5\}} ,\forall n\in [3] \}\\
Z_6 &= \{ W^2_{n,\{1,6\}}, W^2_{n,\{2,6\}}, W^2_{n,\{3,6\}}, W^2_{n,\{4,6\}} ,\forall n\in [3] \}.
\end{align*}

\noindent \textbf{Delivery phase:} Let the demand vector be $(1,2,1,3,2,2)$, and $\mathcal{B}=\{1,2,4\}$.

\noindent \textbf{Server to mirrors:} 
\begin{itemize}
\item The server transmits the following messages:\\ 
   $W^1_{1,2}, W^1_{1,3}\oplus W^1_{1,1},  W^1_{1,4},  W^1_{1,5}, W^1_{1,6},
   W^1_{2,1} W^1_{2,3}, W^1_{2,5}\oplus W^1_{2,2},W^1_{2,4},W^1_{2,6} \oplus W^1_{2,2}, W^1_{3,1}, W^1_{3,2},\\ W^1_{3,3}, W^1_{3,5}, W^1_{3,6}.$ 

\item The server also transmits the following coded messages: $\bigoplus_{s\in S} W^2_{d_s, S\backslash \{s\}}$ for all $S \subseteq [6]$ such that $|S|=3$. 
\end{itemize}
Therefore, the rate is
$$R_1=\frac{15}{2 \times 6}+ \frac{20}{2\times 15}= \frac{115}{60}= 1.917.$$

\noindent \textbf{Mirrors to users:} For $m=1$, we have $D_1=\{1,2\}$ and $t_1=2$. The $1^{\text{st}}$ mirror will transmit the following uncoded and coded messages:		
\begin{itemize}

\item Uncoded messages:
\[
\begin{matrix}
W^1_{1,1},  & W^1_{1,2}, &W^1_{1,3}, &W^1_{1,4},&  W^1_{1,5}, &W^1_{1,6}, \\
W^1_{2,1}, & W^1_{2,1}, & W^1_{2,3}, & W^1_{2,4}, & W^1_{2,5}, & W^1_{2,6}.

\end{matrix}
\]
\item Coded messages: $$\bigoplus_{s\in S} W^2_{d_s, S\backslash \{s\}},$$ for all $S \in \mathcal{S} \backslash \{ \{3,4,5\}, \{3,4,6\}, \{3,5,6\}, \{4,5,6\}\}$, where $\mathcal{S}=\{ T \subseteq [6] \ | \ |T|=3\}$.
\item The $1^{\text{st}}$ mirror also transmits $W^2_{1,\{1,2\}}$ and $W^2_{2,\{1,2\}}$ from its cache.
\end{itemize}
Therefore, the rate is 
$$R^{(1)}_2=\frac{12}{2 \times 6}+ \frac{18}{2\times 15}= \frac{8}{5}=1.6.$$
For $m=2$, we have $D_2=\{1,3\}$ and $t_2=2$. The $2^{\text{nd}}$ mirror will transmit the following uncoded and coded messages:
\begin{itemize}
\item Uncoded messages: 
\[
\begin{matrix}
W^1_{1,1},& W^1_{1,2},&W^1_{1,3}, &W^1_{1,4} &  W^1_{1,5}, &W^1_{1,6}, \\
W^1_{3,1}& W^1_{3,2} & W^1_{3,3}, & W^1_{3,4}, & W^1_{3,5}, &W^1_{3,6}.
\end{matrix}
\]
\item Coded messages: $$\bigoplus_{s\in S} W^2_{d_s, S\backslash \{s\}}$$ for all $S \in \mathcal{S} \backslash \{ \{1,2,5\}, \{1,2,6\}, \{2,5,6\}, \{1,5,6\}\}$, where $\mathcal{S}=\{ T \subseteq [6] \ | \ |T|=3\}$.
\item The mirror also transmits $W^2_{1,\{3,4\}}$ and $W^2_{3,\{3,4\}}$ from its cache.
\end{itemize}
Therefore, the rate is
$$R^{(2)}_2=\frac{12}{2 \times 6}+ \frac{18}{2\times 15}= \frac{8}{5}=1.6.$$
For $m=3$, we have $D_3=\{2\}$ and $t_3=1$.  The $3^{\text{rd}}$ mirror will transmit the following uncoded and coded messages:		
\begin{itemize}
\item Uncoded messages:  
\[
\begin{matrix}
W^1_{2,1}& W^1_{2,2} & W^1_{2,3}, & W^1_{2,4}, & W^1_{2,5},  & W^1_{2,6}.
\end{matrix}
\]
\item Coded messages:  $$\bigoplus_{s\in S} W^2_{d_s, S\backslash \{s\}}$$ for all $S \in \mathcal{S} \backslash \{ \{1,2,3\}, \{1,2,4\}, \{1,3,4\}, \{2,3,4\}\}$, where $\mathcal{S}=\{ T \subseteq [6] \ | \ |T|=3\}$.
\item The mirror also transmits $W^2_{2,\{5,6\}}$ from its cache.
\end{itemize}
Therefore, the rate is
$$R^{(3)}_2=\frac{6}{2 \times 6}+ \frac{17}{2\times 15}= \frac{16}{15}=1.0667.$$
Clearly, $R_2=1.6$ and the composite rate is $\overline{R}=6.717$.

\end{example}

\begin{table*}[ht]
	\centering
	\caption{Comparison of schemes with $K_1=3$, $K_2=2$, $N=3$, and $\overline{M}=3.2 \ (t=2, \alpha=\frac{1}{2})$  for Example \ref{ex2}} 
	\begin{tabular}{| c| c | c |c|c| c| }
		\hline
		\rule{0pt}{2.5ex}
		\makecell{} & \makecell{KNMD Scheme\\\cite{karamchandani2016hierarchical}}  & \makecell{ZWXWL Scheme\\\cite{7763243}} & \makecell{ZWXWLL Scheme\\\cite{8481555}} &\makecell{WWCY Scheme\\\cite{wang2019reduce}}  & \makecell{Proposed scheme \\ (For same $\overline{M}$)}  \\ [2pt]
		\hline
  
		\rule{0pt}{2.5ex}
		 $M_{1}$ &  $0.267$&  $0.967$  &  $0.267$ &  $0.267$ &   $0.267$ 
		\\
		\hline
		\rule{0pt}{2.5ex}
		$M_{2}$ & $0.4$& $0.05$ & $0.4$ & $0.4$   &  $0.4$  \\
		\hline

        \rule{0pt}{2.5ex}
		 $\overline{M}$ & $3.2$& $3.2$  & $3.2$ & $3.2$ &  $3.2$
		\\
		\hline
		\rule{0pt}{2.5ex}
		$R_{1}$ &  $3.076$&  $2.033$ & $3.413$ & $3.07$   &  $1.917$  \\
		\hline
		
		\rule{0pt}{3.5ex}
       $R_{2}$ &$1.623$ &$1.95$&$1.618$ & $1.623$    & $1.6$
		\\
		\hline

       \rule{0pt}{3.5ex}
       $R_{1}+K_{1}R_{2}$ & $7.945$& $7.88$ & $8.267$& $7.939$   &  $6.717$
		\\
		\hline
	\end{tabular}
	
	\label{tab2}
\end{table*}

For Example, \ref{ex2}, the rates of the proposed scheme are compared with the existing schemes in Table \ref{tab2} for the same parameters $N, K_1, K_2$ and $\overline{M}=3.2$. We kept the value of global memory the same for all the existing schemes and computed the value of $R_1$, $R_2$, and the composite rate $\overline{R}$. It is clear from the table that the proposed scheme achieves the lowest composite rate ($\overline{R}=6.717$) for the global memory $\overline{M}=3.2$. The values of $R_1$ and $R_2$ are also minimum for the proposed scheme. The computation of the rates of the existing schemes and the comparison are discussed in detail in the next section.

\section{Comparison with the state-of-the-art} \label{compare}

In this section, we compare the composite rate of our proposed scheme to that of existing schemes with respect to global memory. Additionally, we compare \( R_1 \), \( R_2 \), and \( R_1 + R_2 \) with \( M_1 \) and \( M_2 \).

\subsection{Comparison with the KNMD scheme \cite{karamchandani2016hierarchical}}

Since our scheme covers more memory points $(M_1, M_2)$ for $N \leq K$, we compare it with existing schemes that operate for $N \geq K$ for the case $N=K$. In this comparison, we set $N=K$, where the number of users equals the number of files, to ensure a fair comparison between our scheme and the existing ones. 

Our scheme is designed to operate within the memory range as mentioned in \eqref{global memory}. The KNMD scheme introduces three memory regions and their respective optimal values of $\alpha'$ and $\beta$ \footnote{For the KNMD scheme, we will use $\alpha'$ instead of $\alpha$ to avoid confusion with $\alpha$ used in this paper. Similarly, for the ZWXWLL scheme and the WWCY scheme, we use $\alpha''$ and $\alpha'''$, respectively, instead of $\alpha$.}{[Eq. \ref{optimal}]}. In our scheme, we will establish the criteria for our alpha values to ascertain our placement within a particular memory region for $t=K_2$. In this section, we consider $t=K_2$. By establishing these criteria, we can align our scheme with the corresponding memory region of the KNMD scheme. The KNMD scheme provides the following rates for the first and second layers: 
\

\begin{align}
    R_{1}(\alpha', \beta) &\triangleq \alpha' {{}\cdot{}} K_{2}  {{}\cdot{}} r\left(\frac{M_{1}}{\alpha' N}, K_{1}\right)
     + (1-\alpha'){{}\cdot{}}r\left(\frac{(1-\beta) M_{2}}{(1-\alpha') N}, K_{1} K_{2}\right) , \\
    R_{2}(\alpha', \beta) &\triangleq  \alpha' \cdot r\left(\frac{\beta M_{2}}{\alpha' N}, K_{2}\right)
    + (1-\alpha') \cdot r\left(\frac{(1-\beta) M_{2}}{(1-\alpha') N}, K_{2}\right),
\end{align}
where $\alpha', \beta \in [0,1]$, and
    \begin{small}
   \begin{equation}
    r\left(\frac{M}{N}, K\right) \triangleq\left[K \cdot\left(1-\frac{M}{N}\right) \cdot \frac{N}{K M}\left(1-\left(1-\frac{M}{N}\right)^{K}\right)\right]^{+} \nonumber
\end{equation}
\end{small}

The KNMD scheme considers the following three distinct regions of $M_{1}$ and $M_{2}$,

Region I: $M_{1}+M_{2} K_{2} \geq N$ and $0 \leq M_{1} \leq N / 4$,

Region II: $M_{1}+M_{2} K_{2} < N$, 

Region III: $M_{1}+M_{2} K_{2} \geq N$ and $N / 4<M_{1} \leq N$.
\label{ranges}

and corresponding optimal values of $\alpha'$ and $\beta$ as:

\begin{equation}
\left({\alpha'}^{\star}, \beta^{\star}\right) \triangleq \begin{cases}
\left(\frac{M_{1}}{N}, \frac{M_{1}}{N}\right) & \text { in Region I, } \\ \left(\frac{M_{1}}{M_{1}+M_{2} K_{2}}, 0\right) & \text { in Region II, } \\ \left(\frac{M_{1}}{N}, \frac{1}{4}\right) & \text { in Region III. }\end{cases}\label{optimal}\end{equation}

\begin{figure*}
    \setlength{\abovecaptionskip}{0pt} 
    \centering
    \includegraphics[width=1.3\textwidth, height=0.5\textheight, keepaspectratio]{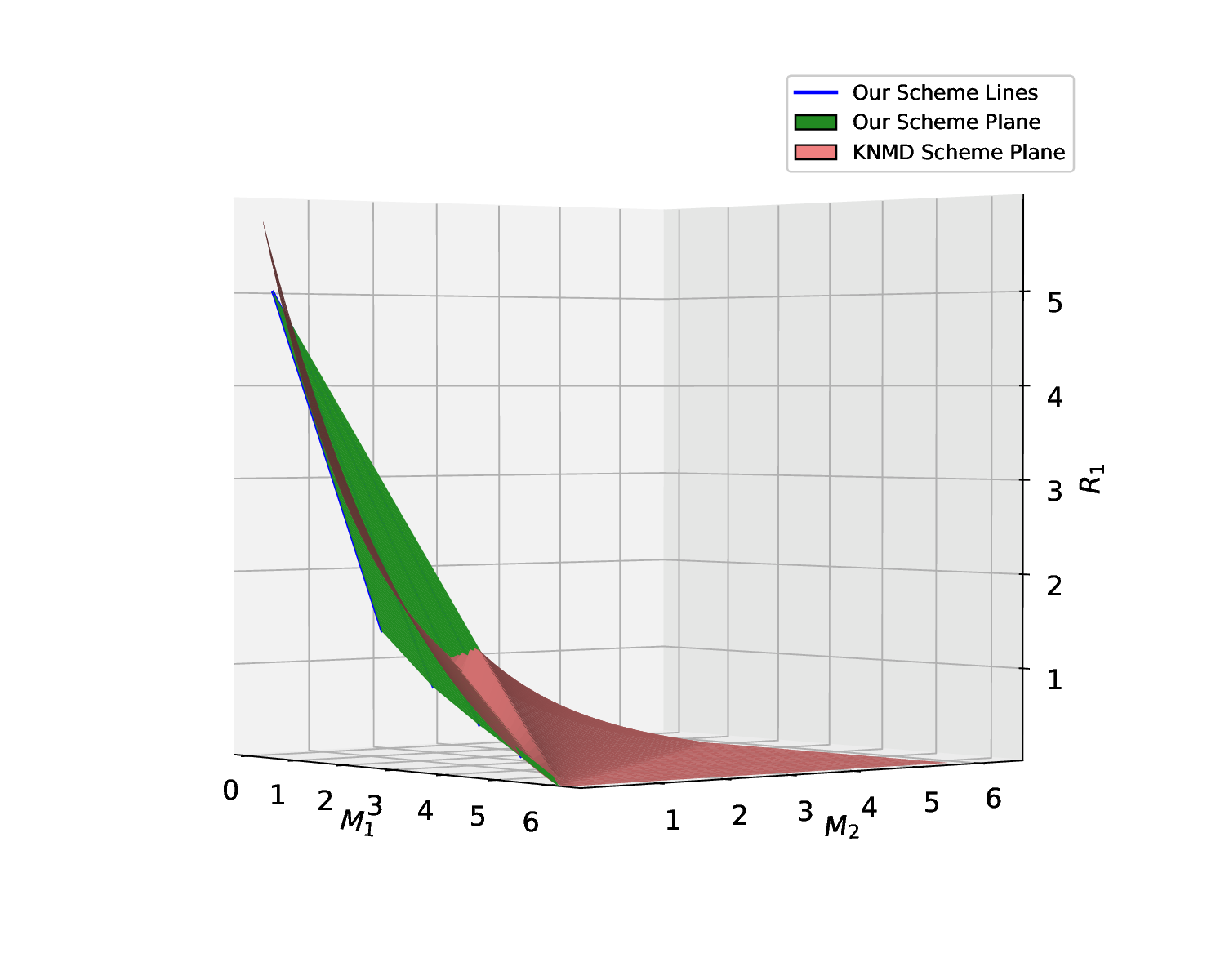}
    \caption{Comparing $R_1$ Trade-offs: KNMD Scheme vs. Our Scheme}
    \label{knmd r1}
    \vspace{0.3cm} 
    \includegraphics[width=1.3\textwidth, height=0.5\textheight, keepaspectratio]{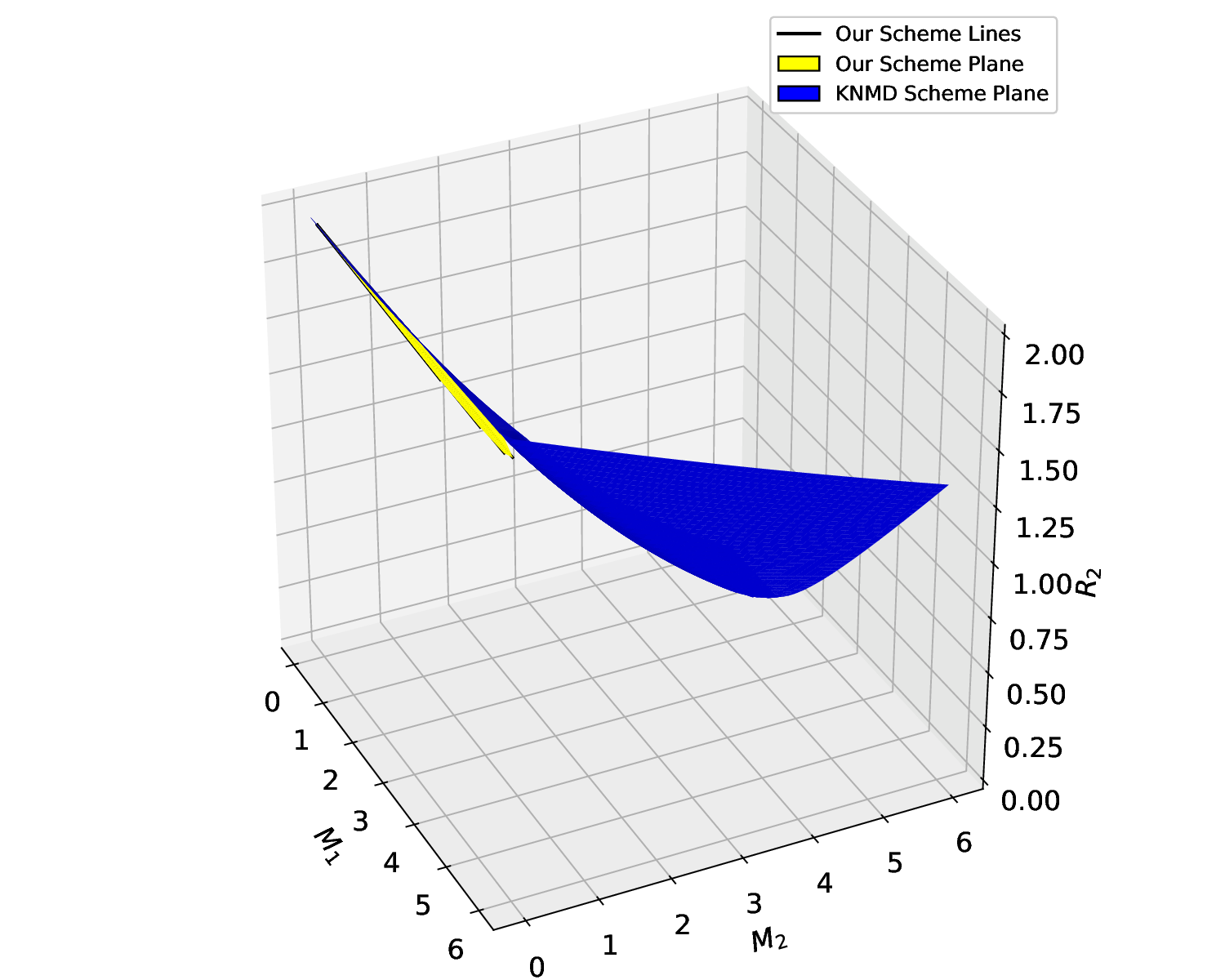}
    \caption{Comparing $R_2$ Trade-offs: KNMD Scheme vs. Our Scheme}
    \label{knmd r2}
\end{figure*}

\begin{figure*}
    \setlength{\abovecaptionskip}{0pt} 
    \centering
    \includegraphics[width=1.3\textwidth, height=0.5\textheight, keepaspectratio]{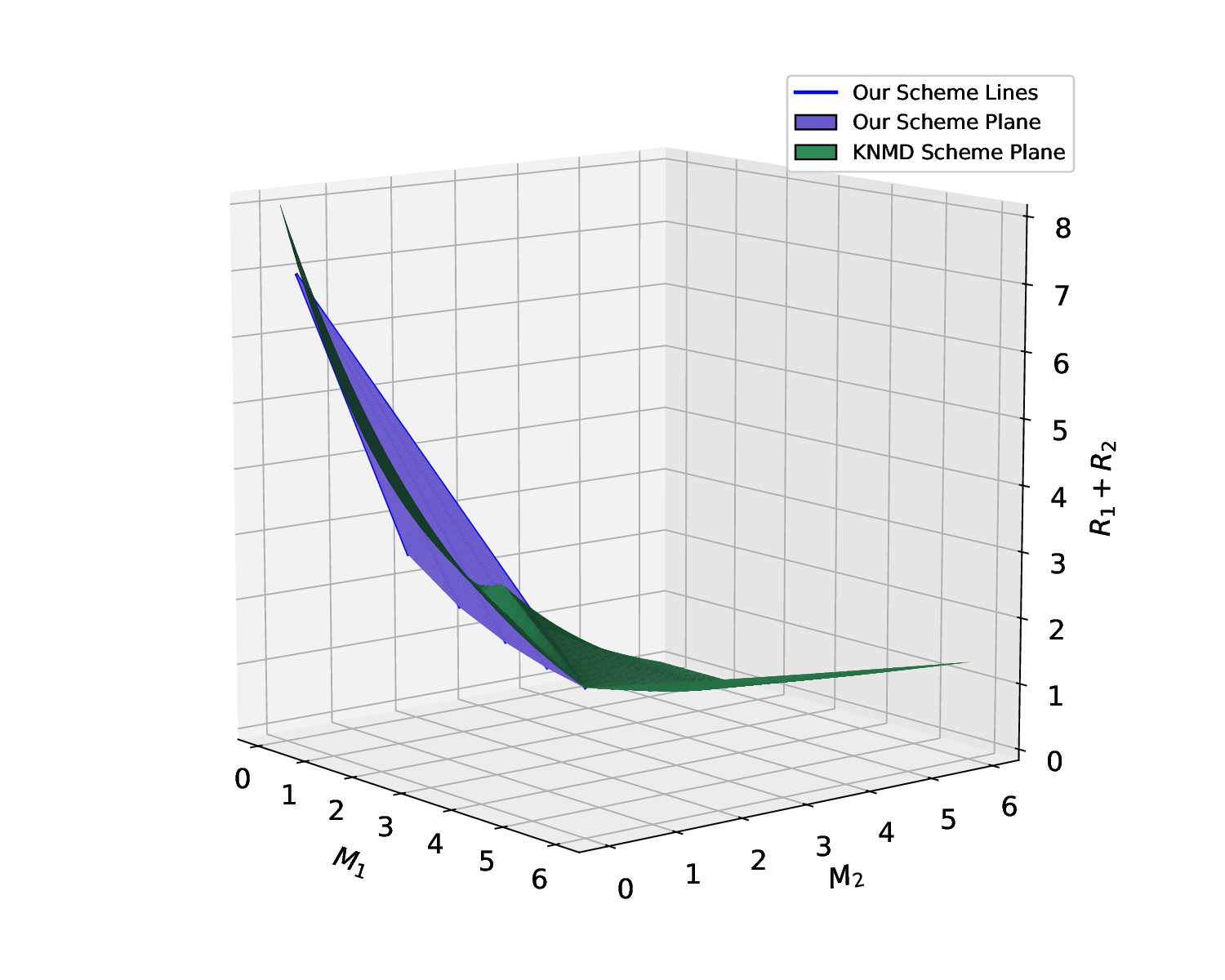}
    \caption{Comparing $R_1$+ $R_2$ Trade-offs: KNMD Scheme vs. Our Scheme}
    \label{knmd r1r2}
    \vspace{0.5cm} 
    \includegraphics[width=1.3\textwidth, height=0.5\textheight, keepaspectratio]{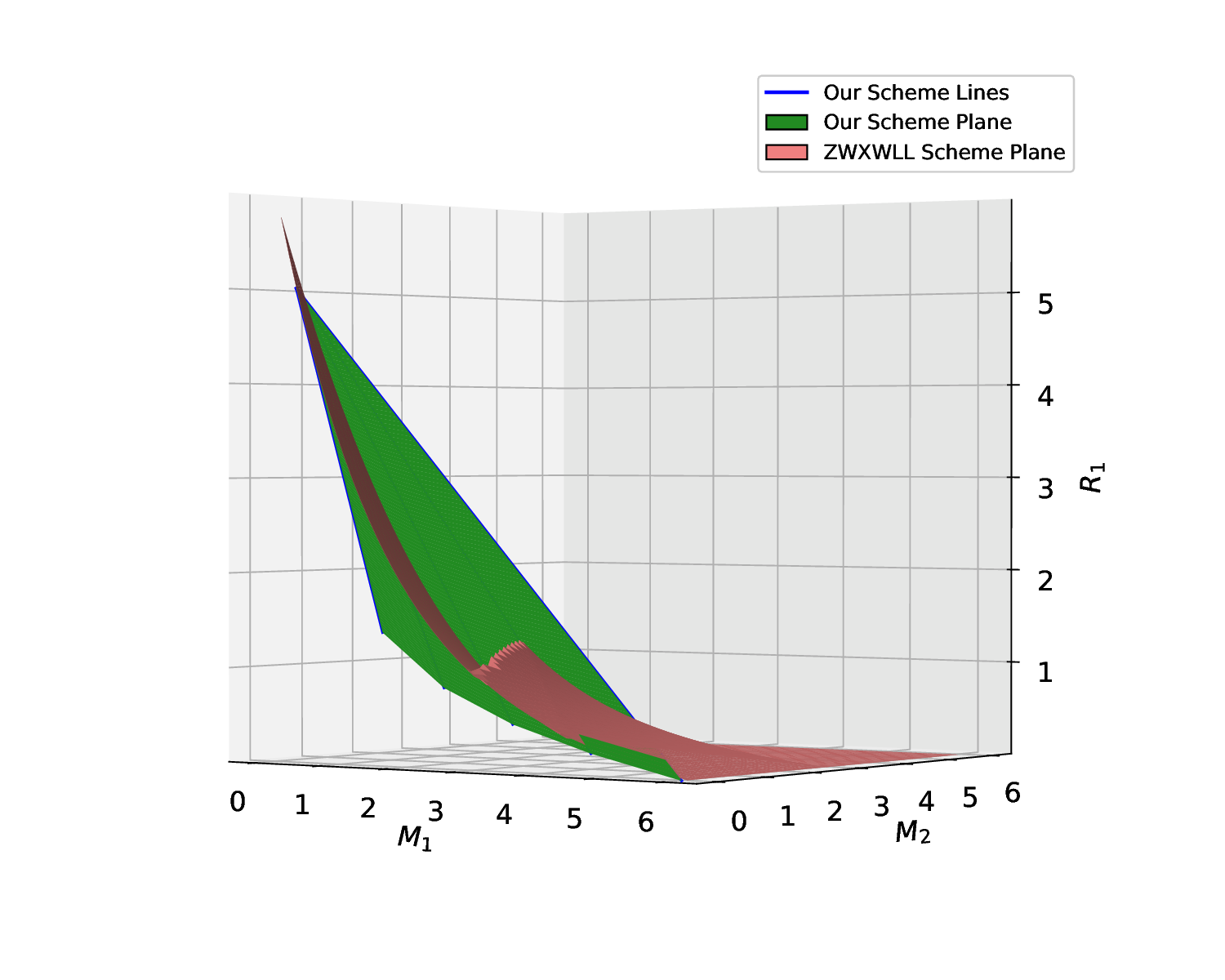}
    \caption{Comparing $R_1$ Trade-offs: ZWXWLL Scheme vs. Our Scheme}
    \label{zwxwll r1}
\end{figure*}

\begin{figure*}
    \setlength{\abovecaptionskip}{0pt} 
    \centering
    \includegraphics[width=1.3\textwidth, height=0.5\textheight, keepaspectratio]{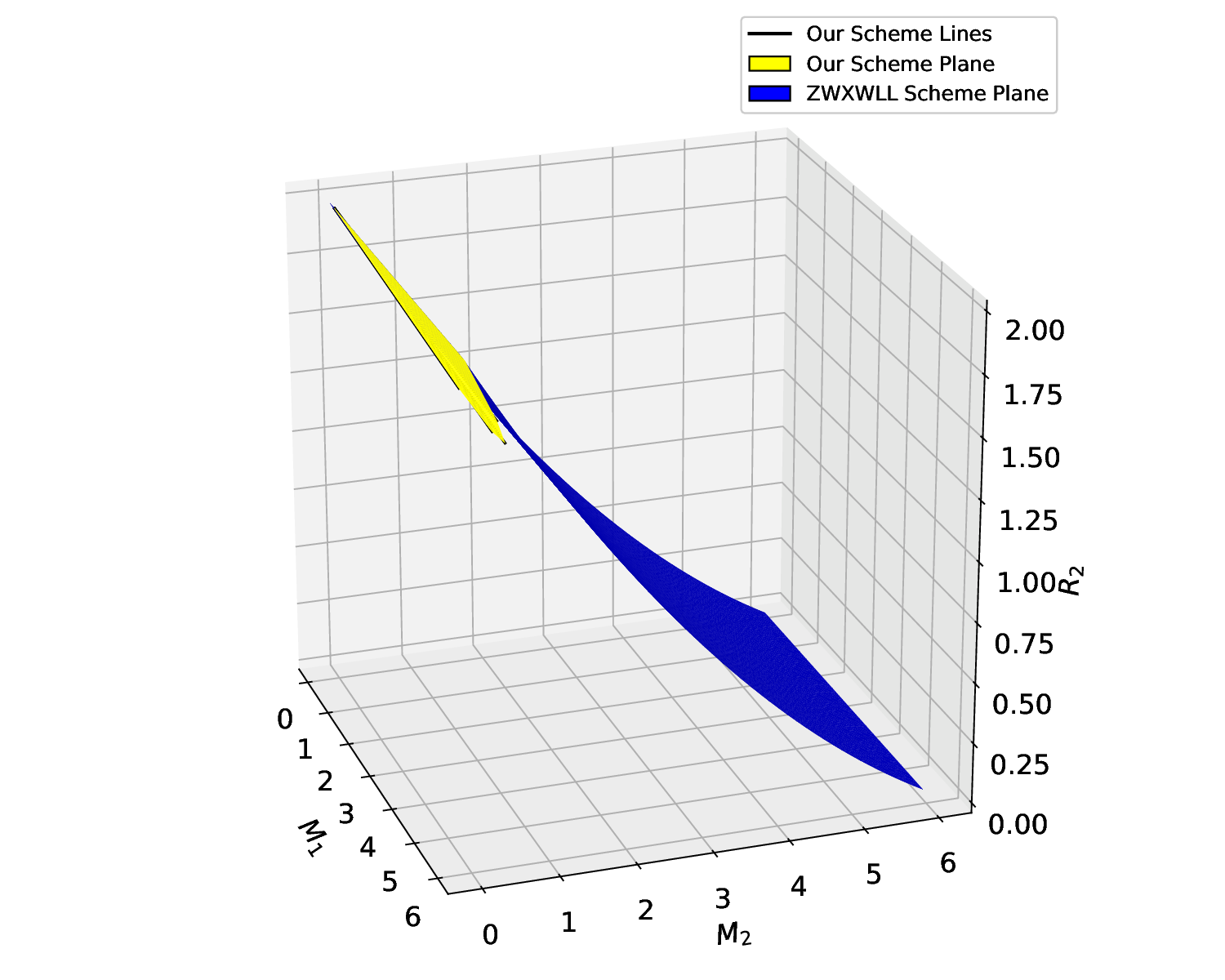}
    \caption{Comparing $R_2$ Trade-offs: ZWXWLL Scheme vs. Our Scheme}
    \label{zwxwll r2}
    \vspace{0.1cm} 
    \includegraphics[width=1.3\textwidth, height=0.5\textheight, keepaspectratio]{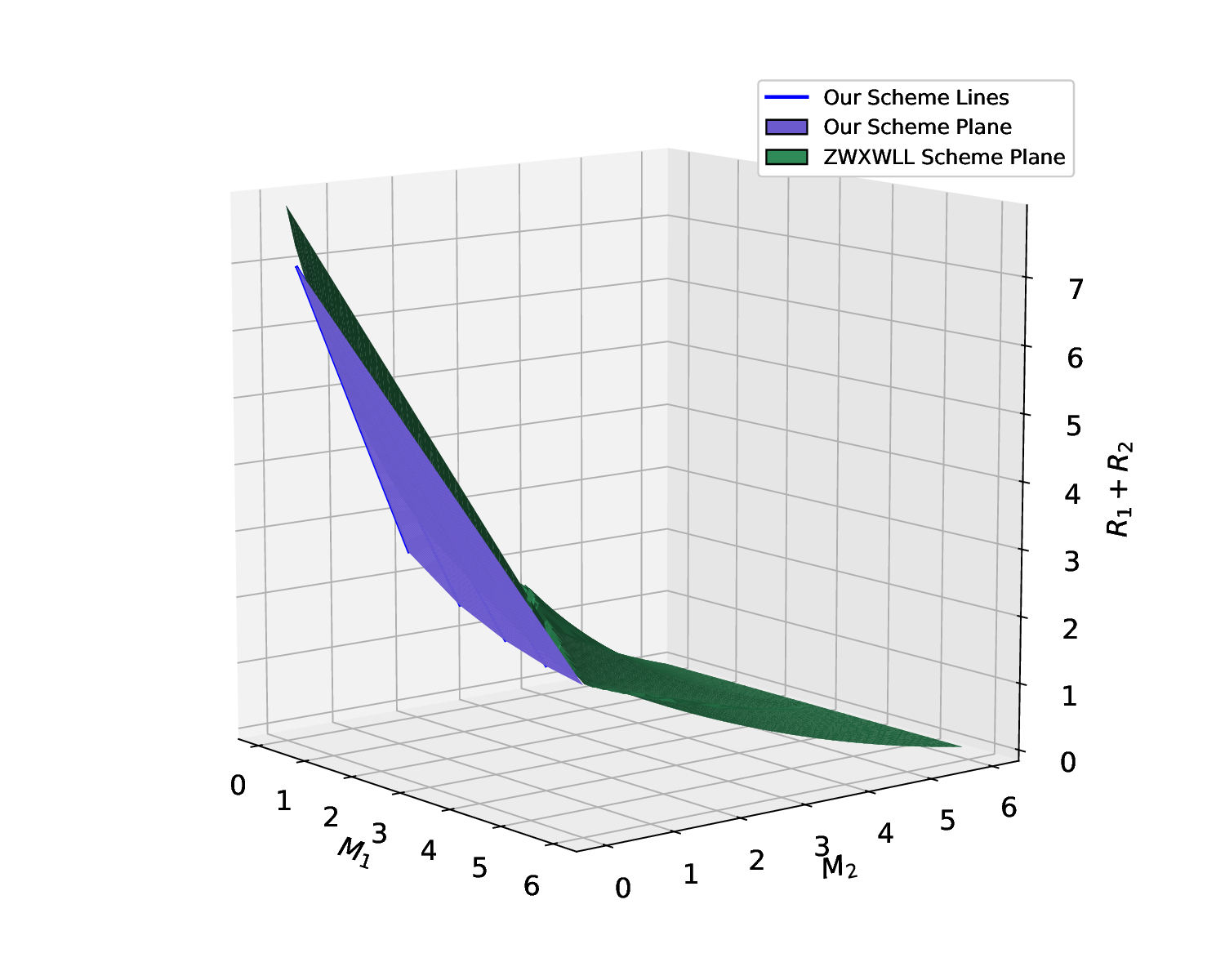}
    \caption{Comparing $R_1$+$R_2$ Trade-offs: ZWXWLL Scheme vs. Our Scheme}
    \label{zwxwll r1r2}
\end{figure*}

\begin{figure*}
    \setlength{\abovecaptionskip}{0pt} 
    \centering
    \includegraphics[width=1.3\textwidth, height=0.5\textheight, keepaspectratio]{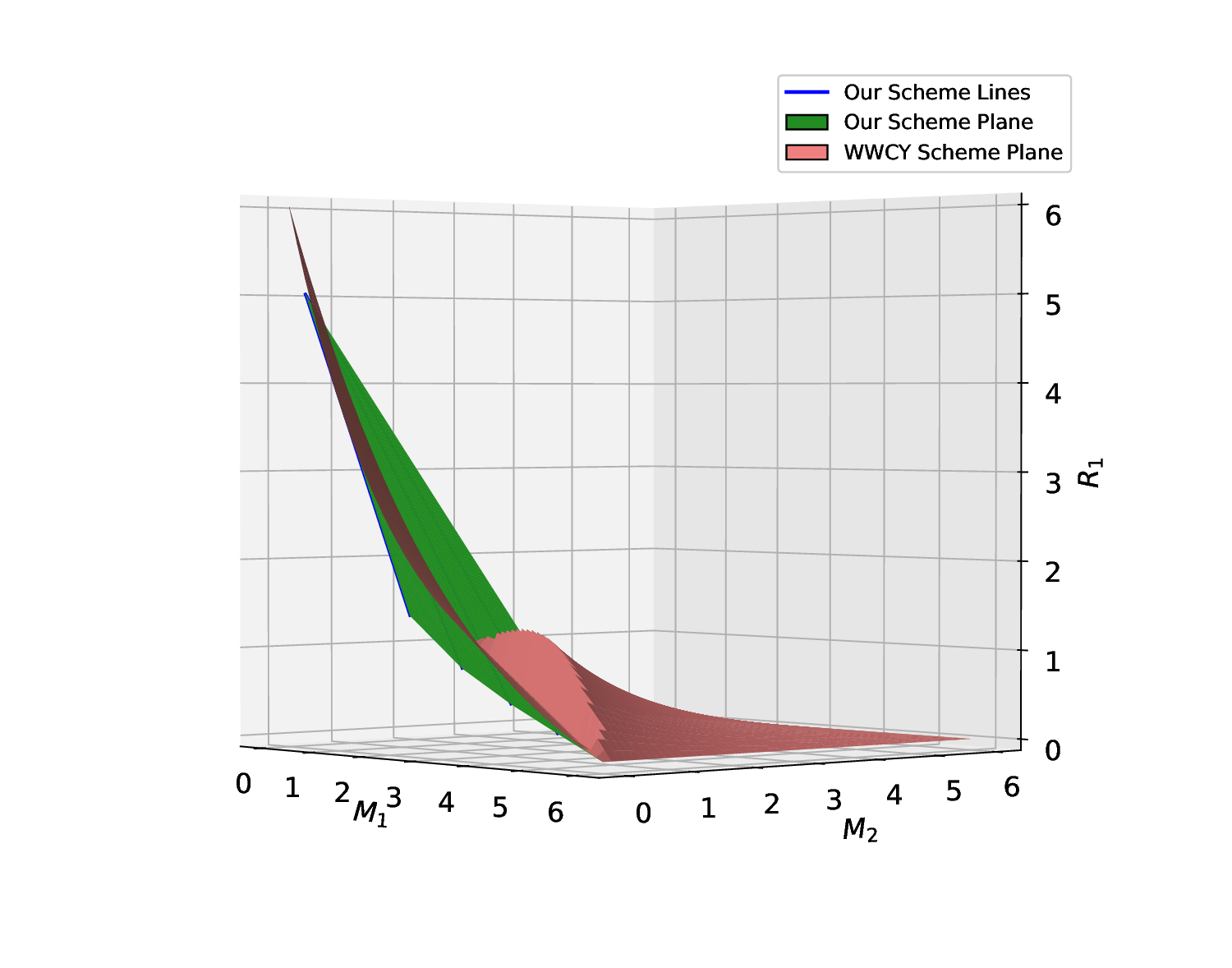}
    \caption{Comparing $R_1$ Trade-offs: WWCY Scheme vs. Our Scheme}
    \label{wwcy r1}
    \vspace{0.1cm} 
    \includegraphics[width=1.3\textwidth, height=0.5\textheight, keepaspectratio]{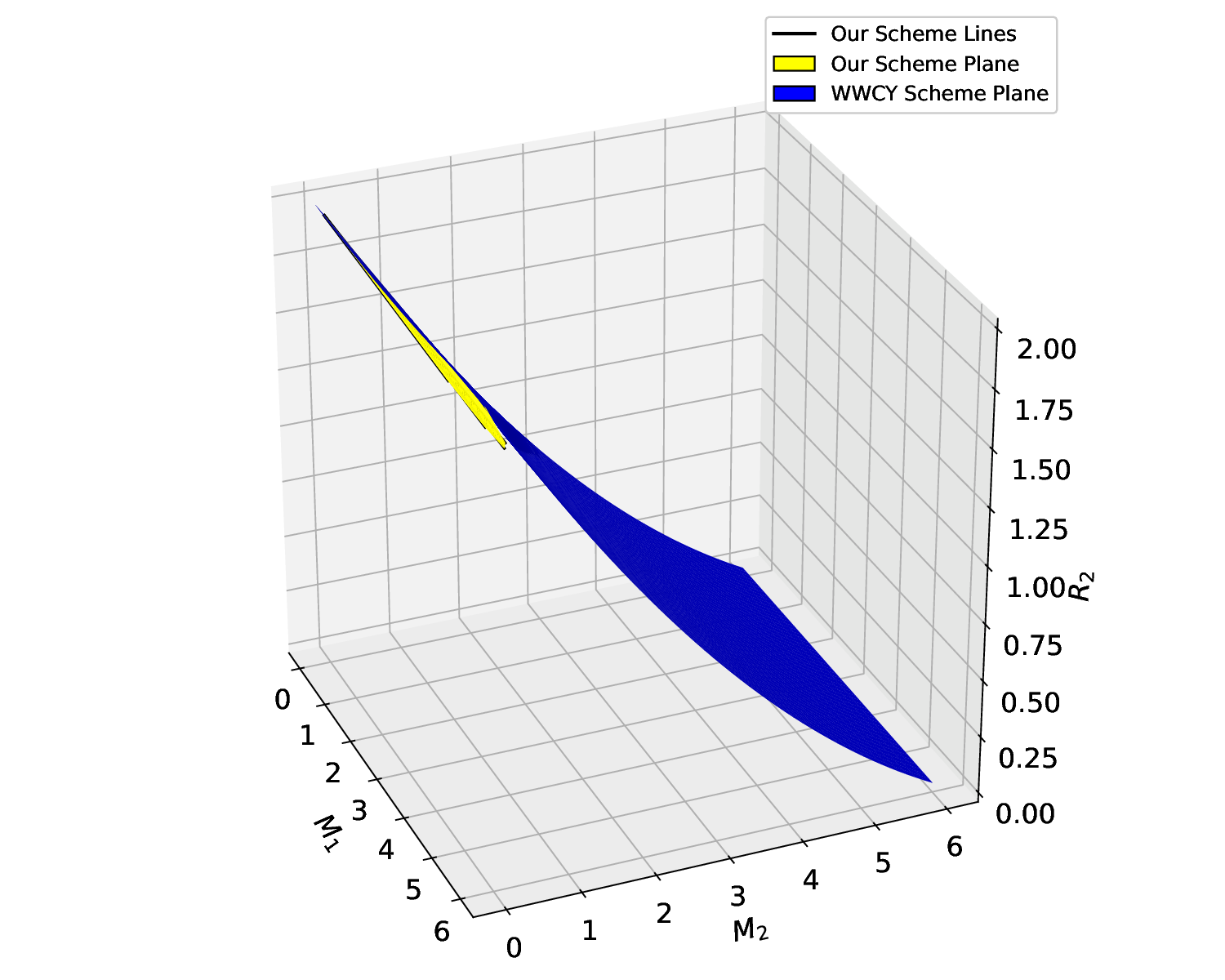}
    \caption{Comparing $R_2$ Trade-offs: WWCY Scheme vs. Our Scheme}
    \label{wwcy r2}
\end{figure*}

\begin{figure*}
    \centering
    \includegraphics[width=1.9\linewidth, height=0.5\textheight, keepaspectratio]{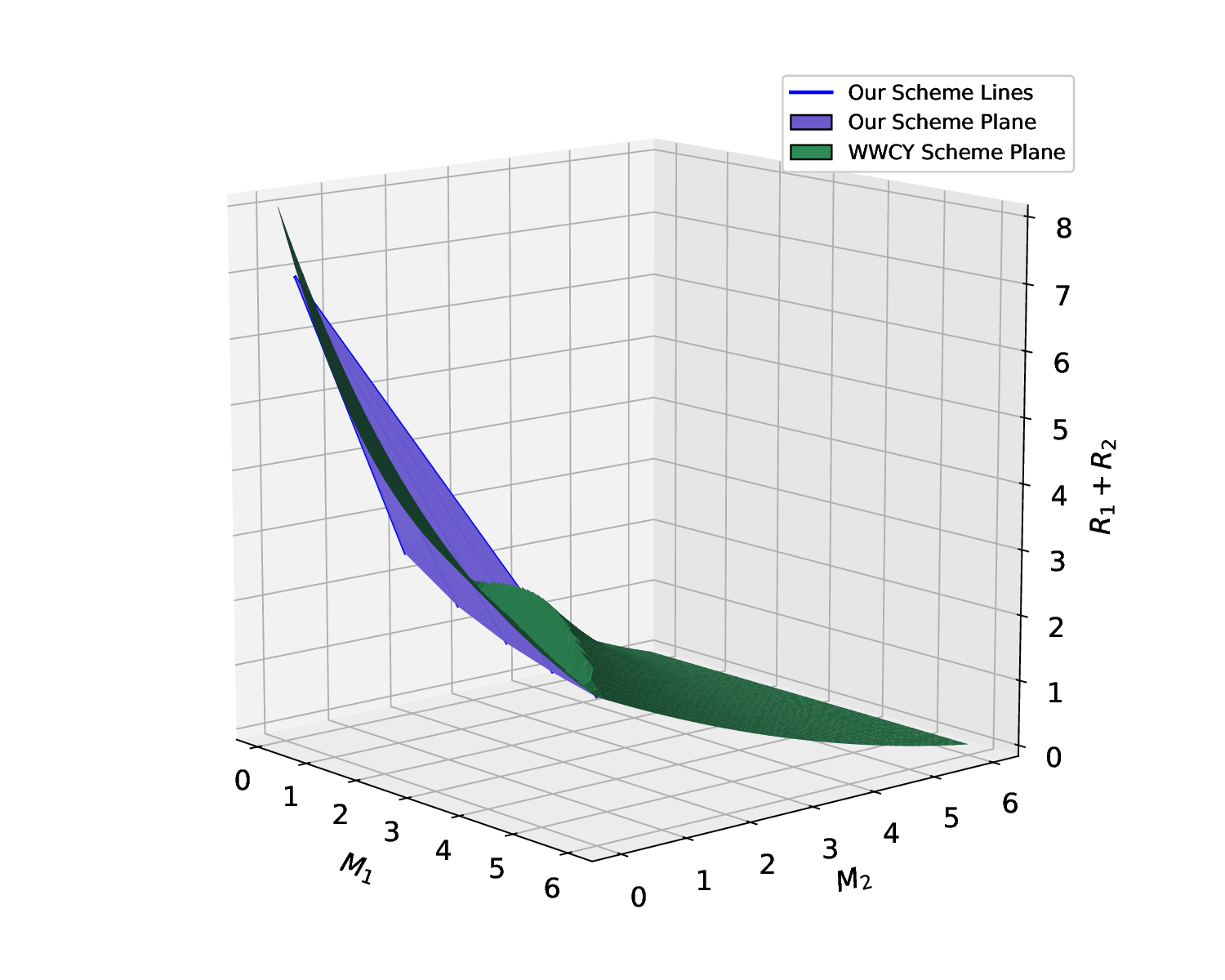}
    \caption{Comparing  $R_1$+$R_2$ Trade-offs: WWCY Scheme vs. Our Scheme}
    \label{wwcy r1r2}
    
\end{figure*}

\begin{thm}
	For a given pair of positive integers, $K_1$ and $K_2$, our scheme is classified into three distinct regions of the KNMD scheme based on the ranges of the variable $\alpha$. These regions are defined as follows:
\begin{equation}
\begin{cases}
\text{Region I:} & \alpha \leq \min\left[\frac{A-K}{A-\frac{1}{K_1}},B\right] \\
\text{Region II:} & \alpha > \frac{A-K}{A-\frac{1}{K_1}} \\
\text{Region III:} & B < \alpha \leq \frac{A-K}{A-\frac{1}{K_1}}\ \, .
\end{cases}
\end{equation}
where $A=K_2^2 - \frac{K(K_2-1)}{{K \choose {K_2}}}$ and $B =\frac{KK_{1}}{4} \left[\frac{{K \choose K_2}-4}{{K \choose K_2}-KK_{1}}\right]$.
\label{thm_region}
\end{thm}
Proof: Theorem \ref{thm_region} is proved in Appendix A.

\begin{thm}\label{thm2}
For integers $K_1 \geq 2$ and $K_2 \geq 2$ such that $K_1 > K_2$, the proposed scheme falls into Region II for $t=K_2$.
\end{thm}

\begin{thm}\label{thm3}
For integers $K_1 \geq 2$ and $K_2 \geq 2$ except $K_1=K_2=2$ such that $K_1 \leq K_2$,  the proposed scheme will not fall into Region III for $t=K_2$.
\end{thm}

The proofs of Theorem \ref{thm2} and \ref{thm3} are given in Appendix A. The following result can be directly obtained from Theorem \ref{thm2} and \ref{thm3}.

\begin{cor}\label{cor3}
For integers $K_1 \geq 2$ and $K_2 \geq 2$ except $K_1=K_2=2$, the proposed scheme either falls into Region I or Region II.
\end{cor} 

Since the expressions of the rates are quite complex, instead of a general comparison of the rates, we compare the composite rate with some examples from Region I and Region II. For Example \ref{ex1} and \ref{ex2} given in Section \ref{scheme1}, we can check that $\frac{A-K}{A-\frac{1}{K_1}} < 0$; therefore, both examples fall into Region II, which can also be verified by Theorem \ref{thm2}. The optimal values of $\alpha'$ and $\beta$ for Region II are $\left(\frac{M_{1}}{M_{1}+M_{2} K_{2}}, 0\right)$. Since we can not take $M=0$ in the expression of $r\left(\frac{M}{N}, K\right)$ as given in \cite{maddah2014decentralized}, we compute the rate of the KNMD scheme by taking a small value of $\beta$, i.e., $\beta=0.01$ instead of taking $\beta=0$ in both the examples.
 
Figure \ref{knmd r1} shows the plot for $R_1$ with $M_1$ and $M_2$. Our scheme includes lines corresponding to different values of $t$, representing achievable rates within specific memory regions defined by Eq. \ref{global memory} with $\alpha = 0$ and $\alpha = 1$ as the extreme values. The variable $t$ ranges from 1 to $K$, and we have plotted lines from $1$ to $K-1$ in the figure.
Using memory sharing, as explained in Section \ref{memory sharing}, we achieve the regions between each consecutive line. Thus, Figure \ref{knmd r1} includes both the lines and the intermediate regions for our scheme, as well as the plane representing the KNMD scheme. Our scheme's intermediate regions achieve lower $R_1$ compared to the KNMD scheme's plane. At higher $t$ values, our lines or the regions between them intersect the KNMD scheme's plane at specific memory points, indicating competitive performance. Choosing lower $t$ values, or the regions around these lower values, typically provides a better trade-off between $R_1$, $M_1$, and $M_2$.

Figure \ref{knmd r2} shows the trade-off between $R_2$ with $M_1$ and $M_2$ for both the KNMD scheme and our proposed scheme. We observe that the lines corresponding to $t=1$ to $K-1$, as well as the regions between them, consistently fall below the plane representing the KNMD scheme. This indicates that our scheme achieves a reduced $R_2$ compared to the KNMD scheme in these memory regions.
Additionally, Figure \ref{knmd r1r2} explores the trade-off for $R_1 + R_2$ versus $M_1$ and $M_2$. Here, we can also observe that in regions where the plane for our scheme falls below the plane for the KNMD scheme, we achieve a lower combined rate $R_1 + R_2$ than the KNMD scheme.
From the comparison given in Figure \ref{fig:CRvsGM} and Figure \ref{R_1+R_2}, we observe that the trade-off for the composite rate and $R_1 + R_2$ with global memory in the proposed scheme is superior to that of the KNMD scheme for Example \ref{ex1}. Also, Table \ref{tab1} and Table \ref{tab2} illustrate the comparison of the rates for the same value of global memory for Example \ref{ex1} and Example \ref{ex2}, respectively. Now, the following example falls in Region I.

\begin{example}\label{ex3}
Consider $K_1=2, K_2=3$ and $N=6$. We have $A=8.4$ and $\frac{A-K}{A-\frac{1}{K_1}}=0.303$. Therefore, if we choose $\alpha<0.303$, it will fall into Region I. If we consider $\alpha=0.2$, by using the proposed scheme for $t=K_2=3$, we get
$$M_1=0.34, M_2=2.16, \overline{M}=13.64,$$
and 
$$R_1=1.6, R_2 = 1.32, \overline{R}=4.24.$$
Whereas using KNMD scheme for the memory point $(M_1, M_2)=(0.34, 2.16)$ at ${\alpha'}^{\star}=\beta^{\star}=\frac{M_1}{N}=0.0567$, we get 
$$R_1=1.56, R_2 = 1.312, \overline{R}=4.189.$$
Our scheme works for all $t\in [K]$. Figure \ref{fig:ex3} shows all the global memory-composite rate points obtained by the proposed scheme for $t=2$ and $t=3$ for this example. The point we obtained above is denoted by $S (13.64, 4.24)$. The points $A (12, 3.867)$ and $B (16.8, 2.55)$ are obtained by the proposed scheme for $t=2, \alpha=0$ and $t=3, \alpha =0$, respectively. Therefore a point on line $AB$ can be obtained corresponding to the global memory $13.64$ which is denoted by $T$ in Figure \ref{fig:ex3}. At point $T$, we have the following parameters,
$$M_1 = 0.105, M_2 = 2.245, \overline{M}=13.68,$$
and
$$R_1=1.134 , R_2=1.142 , \overline{R}=3.42.$$
It can be easily checked that the memory point $(M_1=0.105, M_2= 2.245)$ is in Region I, and for the same memory point, the KNMD scheme attains the following rates,
$$R_1=1.544, R_2=1.2626, \overline{R}=4.069.$$
Therefore, the proposed scheme provides a better rate for the given global memory $\overline{M}=13.68$. Clearly, we can have multiple memory points $(M_1, M_2)$ corresponding to the same value of the global memory $\overline{M}$.
\begin{figure*}[!t]
	\centering
		\includegraphics[width= 0.7\textwidth]{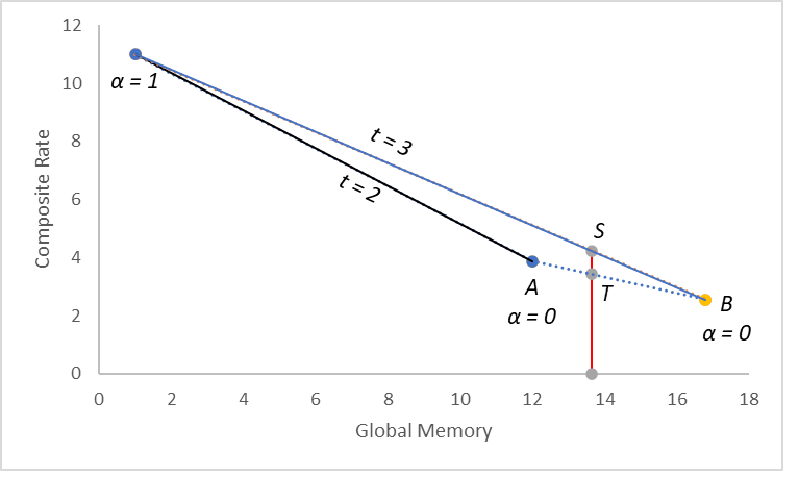}
		\caption{Global memory and composite rate for Example \ref{ex3}}
		\label{fig:ex3}	
\end{figure*} 
\end{example} 

\subsection{Comparison with the ZWXWL scheme\cite{7763243}}

In a two-layer caching network, applying the MN scheme separately in each layer leads to all requested files being retrieved from the server, ignoring user caches. To address this issue, ZWXWL scheme \cite{7763243} introduced a joint caching (JC) scheme that accounts for both layers' caches, reducing unnecessary data delivery from the MN scheme and further lowering the delivery rate. The delivery rates from the server and each mirror of the JC scheme are:
 
\begin{small}
\begin{align}
R_{1}^{J C}\left(M_{1}, M_{2}\right)=K_{1} K_{2}\left(1-\frac{M_{1}}{N}\right)\left(1-\frac{M_{2}}{N}\right) \frac{1}{1+K_{1} M_{1} / N},
\end{align}
\end{small}

\begin{small}
\begin{align}
R_{2}^{J C}\left(M_{2}\right)=K_{2}\left(1-\frac{M_{2}}{N}\right) \frac{1}{1+K_{2} M_{2} / N}.
\end{align}
\end{small}

Since this is a centralized scheme with normalized memory sizes $M_{1} \in$ $\left\{0, N / K_{1}, 2 N / K_{1},\cdots, \right.$ $ \left. N\right\}$ and $M_{2} \in$ $\left\{0, N / K_{2}, 2 N / K_{2}, \cdots, N\right\}$,  we utilized memory sharing to compute the transmission rates at these memory points for comparison with our own scheme.
Figure \ref{fig:CRvsGM} illustrates the trade-off between the composite rate and the global memory for Example \ref{ex1}. We used the global memory-composite rate points $(0, 12), (6, 8)$ and $(18, 4.5)$ which are achieved by the ZWXWL scheme for $(M_1=0, M_2=0), (M_1=2, M_2=0)$ and $(M_1=0, M_2=3)$, respectively, in the plot given in Figure \ref{fig:CRvsGM} for Example \ref{ex1}.

To compute the composite rate of the ZWXWL scheme for the global memory $\overline{M}=3.2$ in Example \ref{ex2}, we used the global memory-composite rate points $(3, 8)$ and $(9, 4.5)$ which are achieved by the scheme for $(M_1=1, M_2=0)$ and $(M_1=0, M_2=1.5)$, respectively. For the same global memory, the comparison of rates of the proposed scheme and the ZWXWL scheme is given in Table \ref{tab1} and Table \ref{tab2} for the Example \ref{ex1} and Example \ref{ex2}, respectively.

\subsection{Comparison with the ZWXWLL scheme\cite{8481555}}
The ZWXWLL scheme\cite{8481555} demonstrates that it achieves a reduction in the transmission rate of the first layer compared to the generalized caching scheme in the KNMD scheme while maintaining the transmission rate of the second layer. The transmission rates for the ZWXWLL scheme are given as follows:

\begin{small}
\begin{align}
  R_{1}^{H}(\alpha'', \beta)= & \alpha'' K_{2}\left(1-\frac{M_{1}}{\alpha'' N}\right)\left(1-\frac{\beta M_{2}}{\alpha'' N}\right) \frac{\alpha'' N}{M_{1}}\nonumber \left(1-\left(1-\frac{M_{1}}{\alpha'' N}\right)^{K_{1}}\right) \\\nonumber
  & +(1-\alpha'')\left(1-\frac{(1-\beta) M_{2}}{(1-\alpha'') N}\right) \frac{(1-\alpha'') N}{(1-\beta) M_{2}} \left(1-\left(1-\frac{(1-\beta) M_{2}}{(1-\alpha'') N}\right)^{K_{1} K_{2}}\right),
\end{align}
\end{small}

\begin{small}
\begin{align}
 R_{2}^{H}(\alpha'', \beta)= & \alpha''\left(1-\frac{\beta M_{2}}{\alpha'' N}\right) \frac{\alpha'' N}{\beta M_{2}}\left(1-\left(1-\frac{\beta M_{2}}{\alpha'' N}\right)^{K_{2}}\right) \nonumber\\
 & +(1-\alpha'')\left(1-\frac{(1-\beta) M_{2}}{(1-\alpha'') N}\right) \frac{(1-\alpha'') N}{(1-\beta) M_{2}} \left(1-\left(1-\frac{(1-\beta) M_{2}}{(1-\alpha'') N}\right)^{K_{2}}\right).
\end{align}
\end{small}

In the ZWXWLL scheme\cite{8481555}, the authors introduced two regimes: Regime I) $M_{1}+K_{2} M_{2} \geq N$ and Regime II) $M_{1}+K_{2} M_{2} < N$. Our approach resides within both of these regions. In the ZWXWLL scheme, multiple tuples of $\alpha''$ and $\beta$ are given for each region, selecting the one with the best rate. We evaluated all possible tuples in both regions and selected the one that yielded the lowest transmission rate for comparison with our scheme. 

For first regime, they considered tuples of $(\alpha'', \beta)$ as follows,
\begin{equation}
(\alpha'', \beta)=
\begin{cases}\left(\frac{M_{1}}{N}, \frac{M_{1}}{N}\right), & \text { Tuple I, } \\ \left(\frac{M_{1}}{M_{1}+K_{2} M_{2}}, 0\right), & \text { Tuple II, } \\ (1,1), & \text { Tuple III. }\end{cases}
\label{zwxwloptimal}
\end{equation}

In the second regime, the considered tuples of $(\alpha'', \beta)$ as follows,
$$
(\alpha'', \beta)= \begin{cases}\left(\frac{M_{1}}{N}, \frac{M_{1}}{N}\right), & \text { Tuple I, } \\ \left(\frac{M_{1}}{N}, \frac{1}{2}\right), & \text { Tuple II. } \end{cases}
$$
Given the resemblance between the memory region of the ZWXWLL scheme and the KNMD scheme, the comparison given in Figure \ref{fig:CRvsGM} follows a similar pattern to that for the KNMD scheme. It can be checked that Example \ref{ex1} and Example \ref{ex2} fall in the second regime given above, and the rate comparison for both the examples are given in Table \ref{tab1} and Table \ref{tab2}, respectively. Now we compute the composite rate of the ZWXWLL scheme for Example \ref{ex3} ($K_1=2, K_2=3, N=6, \overline{M}=13.68$), which falls in the first regime, as follows.
$$R_1=1.5445 , R_2=1.2626 , \overline{R}=4.0697.$$
Clearly, our proposed scheme achieves the better composite rate ($\overline{R}=3.42$)  for this example. 

In Figure \ref{zwxwll r1}, the lines corresponding to different values of $t$ and the memory regions between them, achieved through memory sharing, form a plane representing our scheme. Similarly, the ZWXWLL scheme is depicted as a plane. For certain values of $t$, we observe that the lines and the corresponding regions fall below the plane of the ZWXWLL scheme, indicating memory regions where our scheme achieves a reduced rate for the first layer compared to the ZWXWLL scheme. Additionally, the ZWXWLL scheme exhibits abrupt changes in the rate of the first layer at certain memory points, resulting in sudden spikes in the plane representing the ZWXWLL scheme.

Similarly, akin to the comparison analysis conducted with the KNMD scheme, we depict the trade-off for $R_2$ with $M_1$ and $M_2$, and $R_1 + R_2$ with $M_1$ and $M_2$ in Figure \ref{zwxwll r2} and Figure \ref{zwxwll r1r2}, respectively. In Figure \ref{zwxwll r2}, we observe that the lines corresponding to all values of $t$, and the plane between these lines, fall below the plane of the ZWXWLL scheme within the memory region where our scheme is applicable. The trade-off for $R_1 + R_2$ with $M_1$ and $M_2$ in Figure \ref{zwxwll r1r2} shows that the plane for our scheme lies below the plane for the ZWXWLL scheme, indicating a reduction in $R_1 + R_2$ at these memory points.

\subsection{Comparison with the WWCY scheme\cite{wang2019reduce}}
The WWCY scheme, introduced in \cite{wang2019reduce}, has a lower rate of the first layer compared to the KNMD and ZWXWLL schemes, while the second layer rate remains the same. The scheme defines two regimes: Regime I ($M_{1}+K_{2} M_{2} \geq N$) and Regime II ($M_{1}+K_{2} M_{2}<N$). In the WWCY scheme, the considered tuples of $(\alpha''', \beta)$ are as follows: $$
(\alpha''', \beta)= \begin{cases}\left(\frac{M_{1}}{N}, \frac{M_{1}}{N}\right), & \text { Tuple I, } \\ \left(\frac{M_{1}}{M_{1}+K_{2} M_{2}}, 0\right) & \text { Tuple II. } \end{cases}
$$
The rates of the first and second layers for the WWCY scheme can be expressed as follows:
\begin{small}
 
 \begin{align}
R_{1}(\alpha''', \beta) \triangleq   \alpha'''  \cdot r\left(\frac{M_{1}}{\alpha''' N}, K_{1}\right) r\left(\frac{\beta M_{2}}{\alpha''' N}, K_{2}\right) 
+(1-\alpha''') \cdot r\left(\frac{(1-\beta) M_{2}}{(1-\alpha''') N}, K_{1} K_{2}\right),
\label{wwcyrate1}
 \end{align}

\end{small}

\begin{small}
 \begin{align}
R_{2}(\alpha''', \beta) \triangleq  \alpha''' \cdot r\left(\frac{\beta M_{2}}{\alpha''' N}, K_{2}\right) 
+(1-\alpha''')  \cdot  r\left(\frac{(1-\beta) M_{2}}{(1-\alpha''') N}, K_{2}\right).
\label{wwcyrate2}
 \end{align}
\end{small}

The WWCY scheme also has the same memory regimes as the ZWXWLL scheme, and Example \ref{ex1} and Example \ref{ex2} fall in Regime II while Example \ref{ex3} falls in Regime I. For the optimal values of $(\alpha''', \beta)$ for Regime II, the tuple  $\left(\frac{M_{1}}{M_{1}+M_{2} K_{2}}, 0\right)$ is considered by the WWCY scheme. Again, we computed the rate of the WWCY scheme by taking $\beta=0.01$ instead of taking $\beta=0$ in Example \ref{ex1} and Example \ref{ex2}.
Figure \ref{fig:CRvsGM} illustrates the trade-off between the composite rate and global memory for Example \ref{ex1}, and the comparison of rates of the proposed scheme and the WWCY scheme for the same global memory is given in Table \ref{tab1} and Table \ref{tab2} for the Example \ref{ex1} and Example \ref{ex2}, respectively.
Now we compute the composite rate of the WWCY scheme for Example \ref{ex3} ($K_1=2, K_2=3, N=6, \overline{M}=13.68$),
$$R_1=1.5445 , R_2=1.2626 , \overline{R}=4.0697,$$
which is the same as the rate of the ZWXWLL scheme, and our proposed scheme achieves the better composite rate ($\overline{R}=3.42$). 

Figure \ref{wwcy r1} illustrates the trade-off between $R_1$ with $M_1$ and $M_2$ for both the WWCY scheme and our scheme. Our scheme consistently achieves a better rate, as indicated by the memory region where the lines and the corresponding plane for our scheme lie below the plane of the WWCY scheme. With increasing values of $t$, we observe improved performance only at specific memory points. \\
The trade-off for $R_2$ with $M_1$ and $M_2$ is depicted in Figure \ref{wwcy r2} for both our scheme and the WWCY scheme. Across all values of $t$, our scheme consistently positions its plane below the plane defined by the WWCY scheme, resulting in reduced $R_2$.

 Figure \ref{wwcy r1r2} shows the trade-off for $R_1 + R_2$ with $M_1$ and $M_2$ for both the WWCY scheme and our scheme. We observe that for lower values of $t$, such as $t = 1$ and $t = 2$, our scheme's plane falls below that of the WWCY scheme, thus achieving a lower $R_1 + R_2$. At $t = 3$, our scheme performs better than the WWCY scheme at some memory points but not consistently across all memory regions. Therefore, selecting lower values of $t$ generally results in a reduced $R_1 + R_2$.

\subsection{Comparison with the KWC scheme\cite{kong2022centralized}}
The KWC scheme introduced a hierarchical placement delivery array (HPDA) based approach. In this scheme, the rate for the first layer is superior to that of the KNMD and WWCY schemes when the parameters $(\alpha', \beta)$ in the KNMD scheme and $(\alpha''', \beta)$ in the WWCY scheme are set to $(1, 1)$.
The rates for the first and second layers in the hierarchical placement delivery array (HPDA) based approach of the KWC scheme are given by the following equations
\begin{small}
 \begin{align}
R_1=\frac{K_{1}K_{2}-t}{t+1},
\label{kwcrat}
 \end{align}
\end{small}
\begin{small}
 \begin{align}
R_2=\frac{K_{1}K_{2}-t}{t+1} -\frac{{{K-K_{2}} \choose t+1}}{{K \choose t}}+\frac{{{K-K_{2}} \choose t-K_{2}}K_{2}}{{K \choose t}},
 \end{align}
\end{small}
where $t = \frac{{K_1K_2 (M_1 + M_2)}}{N}$.
As shown in Subsection \ref{alpha0}, the KWC scheme can be obtained by the proposed scheme by fixing $\alpha=0$ and $K_2 < t < K$. 

\section{The single mirror case} \label{scheme2}
In this section, we propose an alternative scheme for the single mirror case in which all the coded placements are stored in users' caches. 
Also, this scheme works for $0 \leq M_2 \leq \frac{1}{K}$ and for any $M_1$ such that $0 \leq M_1 \leq (1-KM_2)N$.

Consider a hierarchical system with one mirror, and all $K$ users are connected to that mirror, i.e., $K_1=1$ and $K=K_2$. The mirror is equipped with the cache memory of size $M_1$ files, and each user is equipped with the cache memory of size $M_2$ files.
We divide each file $W_n$ of size $F$ bits into two subfiles $W_n^1$ and $W_n^2$ of size $F_1$ and $F_2$ bits, respectively, where $F_1=\alpha F$ and $F_2 = (1-\alpha) F$ for some $0 \leq \alpha \leq 1$, i.e.,
$$W_n \rightarrow W^1_n \quad \text{and} \quad W^2_n, \quad \forall n \in N.$$
We further divide $W_n^1$ into $K$ mini-subfiles, say $W^1_{n,1}, W^1_{n,2}, \ldots, W^1_{n,K}$ for all $n \in [N]$. The size of $W^1_{n,i}$ is $\frac{\alpha F}{K}$ for all $n\in [N], i \in [K]$. This scheme works for the following values of $M_1$ and $M_2$,
$$0 \leq M_1 \leq (1-\alpha) N \quad \text{and} \quad M_2=\frac{\alpha}{K},$$
where $\alpha \in [0,1]$.
\subsection{Placement phase}
Let the cache content of the mirror be denoted by $\Lambda$, and the cache content of user $k$ be denoted by $Z_k$ for $k \in [K]$. 
Then,
\begin{align*}
    \Lambda &= \{ W^{2,\theta}_n \ | \ \forall n\in [N] \}  \quad \text{and}
    \quad
    Z_k = \{ W^1_{1,k} \oplus W^1_{2,k} \oplus \cdots \oplus W^1_{N,k} \},
\end{align*}
where $W_{n}^{2, \theta}$ denotes the $\frac{M_1}{N(1-\alpha)}$ part of subfile $W_n^2$.
Clearly, we have
$$M_1=\frac{N \left(\frac{M_1}{(1-\alpha) N} (1-\alpha) F\right)}{F} \quad \text{and} \quad M_2=\frac{\alpha}{K}.$$
The global memory of the system is
$\overline{M}=M_1+K M_2=M_1+\alpha,$ where $0 \leq M_1 \leq (1-\alpha)^2 N$.

\subsection{Delivery phase}
 Let the demand vector be $\overline{d}=(d_1,d_2, \ldots, d_K)$ and each file is demanded by at least one user.  Consider a base set $\mathcal{B} \subseteq [K]$ such that $|\mathcal{B}|=N$ and $ \{d_k \ | \ k \in \mathcal{B}\} = [N]$.
Then transmissions are as follows:

\noindent \textbf{Server to mirrors:}
The following transmissions are from the server to the mirror.
\begin{enumerate}[label=(\subscript{\textbf{SM}}{{\arabic*}})]
\item For each $k\in [K]$, transmit mini-subfiles $W^1_{j,k}$, for $j\in [N]$ and $j \neq d_k$.
\item For each $k \in [K] \backslash \mathcal{B}$, transmit $W^1_{d_k, k} \oplus W^1_{d_{k},k'}$, where $k' \in \mathcal{B}$ and $d_k=d_{k'}$.
\item For each $n \in [N]$, transmit $W_{n}^{2, \not\theta}$, where $W_{n}^{2, \not\theta}$ denoted the remaining $\left( 1-\frac{M_1}{N(1-\alpha)} \right)$ part of subfile $W_n^2$. 
\end{enumerate}

\noindent \textbf{Mirror to users:}
The following transmissions are from the mirror to the users attached to it.
\begin{enumerate}[label=(\subscript{\textbf{MU}}{{\arabic*}})]
\item For each $k\in [K]$, transmit mini-subfiles $W^1_{j,k}$, for $j\in [N]$ and $j \neq d_k$.
\item For each $k \in [K] \backslash \mathcal{B}$, transmit $W^1_{d_k, k} \oplus W^1_{d_{k},k'}$, where $k' \in \mathcal{B}$ and $d_k=d_{k'}$.
\item For each $n \in [N]$, transmit $W_{n}^{2}$.
\end{enumerate}
\subsection{Proof of correctness}
\subsubsection {Mirror gets from the server what it transmits to users:}
In this part, we prove that after receiving the transmissions from the server, the mirror will get all the files that it needs to transmit to the users in the delivery phase.

Clearly, for steps $(\textbf{MU}_1)$ and $(\textbf{MU}_2)$, the mirror gets the required files directly from steps $(\textbf{SM}_1)$ and $(\textbf{SM}_2)$.
Since for each file $W_n$, $n \in [N]$, the mirror gets $\frac{M_1}{N(1-\alpha)}$ part from its cache and remaining $\left( 1-\frac{M_1}{N(1-\alpha)} \right)$ part from the transmission of step $(\textbf{SM}_3)$. Therefore, the mirror has all the required files for step $(\textbf{MU}_3)$.  \\

\subsubsection {All users get their desired files:}
In this part, we prove that after receiving the transmissions from the mirror, each user will get the demanded file in the delivery phase.
Let the set $D^{(i)}$ contains all the users with demand $W_i$ for all $i \in [N]$, i.e., 
$$D^{(i)}=\{k \in [K] \ | \ d_k=i\}.$$
Consider a user $\lambda \in [K]$. The demand of user $\lambda$ is $W_{d_{\lambda}}$. 
From step $(\textbf{MU}_1)$, the user receives $W^1_{i,\lambda}$ for all $i \in [N]$ and $i \neq d_{\lambda}$. Hence using the cache content $Z_{\lambda}=\{ W^1_{1,\lambda} \oplus W^1_{2,\lambda} \oplus \cdots \oplus W^1_{N,\lambda}\}$ of user $\lambda$, we get $W^1_{d_{\lambda},\lambda}$. Also, from step $(\textbf{MU}_1)$, the user gets $W^1_{d_{\lambda}, j}$ for all $j \in [K]$ such that $d_j \neq d_{\lambda}$. Now, the following mini-subfiles are left to be obtained
$$W^1_{d_{\lambda}, j}, \ \forall j \in D^{(d_{\lambda})} \backslash \{\lambda\}.$$
There are following two cases:
\begin{enumerate}
\item If $\lambda \in \mathcal{B}$, then $D^{(d_{\lambda})} \backslash \{\lambda\} \subseteq  [K]\backslash \mathcal{B}$. Hence from step $(\textbf{MU}_2)$, the user $\lambda$ receives 
$$W^1_{d_{j}, j} \oplus W^1_{d_{j}, \lambda} = W^1_{d_{\lambda}, j} \oplus W^1_{d_{\lambda}, \lambda},$$
 for all $j \in D^{(d_{\lambda})} \backslash \{\lambda\}$. Since the user already has $W^1_{d_{\lambda}, \lambda}$, it will get $W^1_{d_{\lambda}, j}$. 
\item  If $\lambda \not\in \mathcal{B}$, then there exist $\lambda' \in D^{(d_{\lambda})}$ such that $\lambda' \in \mathcal{B}$. Since $\lambda \in [K] \backslash \mathcal{B}$, from step $(\textbf{MU}_2)$, the user receives 
$W^1_{d_{\lambda}, \lambda} \oplus W^1_{d_{\lambda}, \lambda'}$. The user $\lambda$ will get $W^1_{d_{\lambda}, \lambda'}$ as it already has $W^1_{d_{\lambda}, \lambda}$. Again from step $(\textbf{MU}_2)$, the user receives 
$$W^1_{d_{j}, j} \oplus W^1_{d_{j}, \lambda'} = W^1_{d_{\lambda}, j} \oplus W^1_{d_{\lambda}, \lambda'},$$
 for all $j \in D^{(d_{\lambda})} \backslash \{\lambda, \lambda'\}$. Since now the user $\lambda$ already has $W^1_{d_{\lambda}, \lambda'}$, it will get $W^1_{d_{\lambda}, j}$. \\
\end{enumerate}

\subsection{Rate}
\noindent \textbf{Server to mirrors:}
As per the delivery phase, the following rate is calculated step-wise.
\begin{enumerate}
\item The total number of mini-subfiles of size $F_1$ transmitted in step $(\textbf{SM}_1)$ and $(\textbf{SM}_2)$ are $K(N-1) + K-N = N(K-1)$.
\item The size of the subfiles transmitted in step $(\textbf{SM}_3)$ is $N \left( 1-\frac{M_1}{N(1-\alpha)} \right) (1-\alpha) F$.
\end{enumerate}
Therefore, we have the rate, 
\begin{align*}
R_1 &= \alpha \frac{N(K-1)}{K}+N \left( 1-\frac{M_1}{N(1-\alpha)} \right) (1-\alpha) \\
&= \alpha \frac{N(K-1)}{K}+ (1-\alpha) N-M_1 \\
&= N \left( 1- \frac{\alpha}{K} \right)- M_1.
\end{align*}

\noindent \textbf{Mirror to users:}
As per the delivery phase, the following rate is calculated step-wise.
\begin{enumerate}
\item The total number of mini-subfiles of size $F_1$ transmitted in step $(\textbf{MU}_1)$ and $(\textbf{MU}_2)$ are $K(N-1) + K-N = N(K-1)$.
\item The size of the subfiles transmitted in step $(\textbf{MU}_3)$ is $N (1-\alpha) F$.
\end{enumerate}
Therefore, we have the rate, 
$$
R_2 = \alpha \frac{N(K-1)}{K}+N  (1-\alpha) = N \left( 1- \frac{\alpha}{K} \right).
$$
The composite rate of the system is $$\overline{R}=R_1+R_2=2N \left( 1- \frac{\alpha}{K} \right)- M_1.$$

Again, we consider two examples, one for the case $N=K$ and the other for the case $N < K$, given as follows.
\begin{example} \label{ex4}
Let $K_1=1, K_2=4, N=4,$ $\alpha = \frac{1}{2}$. Therefore, we have, 
$$0 \leq M_1 \leq \frac{N}{2}=2 ,  \quad M_2=\frac{\alpha}{K} = \frac{1}{8}.$$
Divide each file $W_n$ into two subfiles $W^1_n$ and $W^2_n$, for all $n \in [4]$. Further, divide $W^1_n$ into $4$ mini-subfiles, say, $W^1_{n,1}, W^1_{n,2}, \ldots, W^1_{n,4}$. 

\noindent \textbf{Placement phase:} The cache contents of the mirror is,
$$\Lambda=\{ W^{2,\theta}_n \ | \ \forall n\in [4] \},$$
where $W_{n}^{2, \theta}$ denoted the $\frac{M_1}{2}$ part of subfile $W_n^2$.
The cache contents of the users are as follows,
\begin{align*}
Z_1 &=\{ W^1_{1,1} \oplus W^1_{2,1} \oplus W^1_{3,1} \oplus W^1_{4,1}\},\\
Z_2 &=\{ W^1_{1,2} \oplus W^1_{2,2} \oplus W^1_{3,2} \oplus W^1_{4,2}\},\\
Z_3 &=\{ W^1_{1,3} \oplus W^1_{2,3} \oplus W^1_{3,3} \oplus W^1_{4,3}\},\\
Z_4 &=\{ W^1_{1,4} \oplus W^1_{2,4} \oplus W^1_{3,4} \oplus W^1_{4,4}\}.
\end{align*}

\noindent \textbf{Delivery phase:} Let the demand vector be $(1,2,3,4)$.

\noindent \textbf{Server to mirrors:} 
\begin{itemize}
\item The server transmits the following files to the mirror:\\
$ W^1_{1,2}, W^1_{1,3}, W^1_{1,4}, 
W^1_{2,1},  W^1_{2,3},  W^1_{2,4}, 
W^1_{3,1},  W^1_{3,2},  W^1_{3,4}, 
W^1_{4,1},  W^1_{4,2},  W^1_{4,3}$.
\item   The server also transmits $W_{n}^{2, \not\theta}$, for each $n \in [4]$, where $W_{n}^{2, \not\theta}$ denoted the remaining $\left( 1-\frac{M_1}{2} \right)$ part of subfile $W_n^2$. 
\end{itemize}
Therefore, we have the rate 
$R_1=\frac{7}{2}-M_1.$
\\
\noindent \textbf{Mirrors to users:}
\begin{itemize}
\item The mirror will transmit the following files:	

$W^1_{1,2}, W^1_{1,3}, W^1_{1,4}, 
W^1_{2,1},  W^1_{2,3},  W^1_{2,4}, 
W^1_{3,1},  W^1_{3,2},    W^1_{3,4}, 
W^1_{4,1},  W^1_{4,2},  W^1_{4,3}$.
\end{itemize}
\begin{itemize}
\item  The mirror also transmits $W_{n}^{2}$, for each $n \in [4]$.
\end{itemize}
Therefore, we have the rate 
$$R_2= \frac{7}{2}.$$
 In this example, the global cache memory is $\overline{M}=M_1+\frac{1}{2}$ and the composite rate is $\overline{R}=7-M_1$. For the same value of $N, K_1$ and $K_2$, the scheme given in section \ref{scheme1} gives the composite rate $\overline{R}=\frac{11}{2}=5.5$ for the global memory $\overline{M}=\frac{5}{2}$.
 For comparison, take $M_1=2$, then the alternate scheme gives the composite rate $\overline{R}=5$ for the global memory $\overline{M}=\frac{5}{2}$.
\end{example}

\begin{example}\label{ex5}
Let $K_1=1, K_2=6, N=4,$ $\alpha = \frac{2}{3}$. Therefore, we have, 
$$0 \leq M_1 \leq \frac{4}{3} ,  \quad M_2=\frac{\alpha}{K} = \frac{1}{9}.$$
Divide each file $W_n$ into two subfiles $W^1_n$ and $W^2_n$, for all $n \in [4]$. Further, divide $W^1_n$ into $6$ mini-subfiles, say, $W^1_{n,1}, W^1_{n,2}, \ldots, W^1_{n,6}$. 

\noindent \textbf{Placement phase:} The cache contents of the mirror is,
$$\Lambda=\{ W^{2,\theta}_n \ | \ \forall n\in [4] \},$$
where $W_{n}^{2, \theta}$ denoted the $\frac{3M_1}{4}$ part of subfile $W_n^2$.
The cache contents of the users are as follows,
\begin{align*}
Z_1 &=\{ W^1_{1,1} \oplus W^1_{2,1} \oplus W^1_{3,1} \oplus W^1_{4,1}\},\\
Z_2 &=\{ W^1_{1,2} \oplus W^1_{2,2} \oplus W^1_{3,2} \oplus W^1_{4,2}\},\\
Z_3 &=\{ W^1_{1,3} \oplus W^1_{2,3} \oplus W^1_{3,3} \oplus W^1_{4,3}\},\\
Z_4 &=\{ W^1_{1,4} \oplus W^1_{2,4} \oplus W^1_{3,4} \oplus W^1_{4,4}\},\\
Z_5 &=\{ W^1_{1,5} \oplus W^1_{2,5} \oplus W^1_{3,5} \oplus W^1_{4,5}\},\\
Z_6 &=\{ W^1_{1,6} \oplus W^1_{2,6} \oplus W^1_{3,6} \oplus W^1_{4,6}\}.
\end{align*}

\noindent \textbf{Delivery phase:} Let the demand vector be $(1,2,2,3,1,4)$ and $\mathcal{B}=\{1,2,4,6\}$. 
\\
\noindent \textbf{Server to mirrors:} 
\begin{itemize}
\item The server transmits the following files to the mirror:

$W^1_{1,2},W^1_{1,1}\oplus W^1_{1,5}, W^1_{1,3}, W^1_{1,4},  W^1_{1,6},W^1_{2,1}, W^1_{2,2} \oplus W^1_{2,3}, W^1_{2,4}, W^1_{2,5}, W^1_{2,6} ,W^1_{3,1}, W^1_{3,2},  W^1_{3,3}, \\ W^1_{3,5},  W^1_{3,6}, W^1_{4,1}, W^1_{4,2},  W^1_{4,3}, W^1_{4,4}, W^1_{4,5}$

\item   The server also transmits $W_{n}^{2, \not\theta}$, for each $n \in [4]$, where $W_{n}^{2, \not\theta}$ denoted the remaining $\left( 1-\frac{3M_1}{4} \right)$ part of subfile $W_n^2$. 
\end{itemize}
Therefore, we have the rate 
$$R_1=\frac{32}{9}-M_1.$$
\noindent \textbf{Mirrors to users:} 	
\begin{itemize}

\item The mirror will transmit the following files: 
\[
\begin{aligned}
& W^1_{1,2}, \, W^1_{1,3}, \, W^1_{1,4}, \, W^1_{1,1} \oplus W^1_{1,5}, \, W^1_{1,6}, \,
W^1_{2,1}, \, W^1_{2,2} \oplus W^1_{2,3}, \, W^1_{2,4}, \, W^1_{2,5}, \, W^1_{2,6}, \\
& W^1_{3,1}, \, W^1_{3,2}, \, W^1_{3,3}, \, W^1_{3,5}, \, W^1_{3,6}, \,
W^1_{4,1}, \, W^1_{4,2}, \, W^1_{4,3}, \, W^1_{4,4}, \, W^1_{4,5}.
\end{aligned}
\]

\item  The mirror also transmits $W_{n}^{2}$, for each $n \in [4]$.
\end{itemize}
Therefore, we have the rate 
$$R_2= \frac{32}{9}.$$

 For the same value of $N, K_1$ and $K_2$, the scheme given in Section \ref{scheme1} gives the composite rate $\overline{R}=\frac{56}{9}=6.22$ for the global memory $\overline{M}=2$.
In this example, the global cache memory is $\overline{M}=M_1+\frac{2}{3}$ and the composite rate is $\overline{R}=\frac{64}{9}-M_1$. For comparison, take $M_1=\frac{4}{3}$, then the alternate scheme gives the composite rate $\overline{R}=\frac{52}{9}=5.78$ for the global memory $\overline{M}=2$.
\end{example}

Now, we show that this scheme is better than the scheme given in Section \ref{scheme1} for the single mirror case.
For $K_1=1$, we have the following parameters from the scheme given in Section \ref{scheme1},
$$\overline{M}=\alpha+(1-\alpha) N \quad \text{and} \quad R=\alpha N \left( 1- \frac{1}{K}\right) + N .$$
The parameters of the alternative scheme presented in this section are as follows.
$$\overline{M'}=M_1+\alpha \quad \text{and} \quad R'=2N \left( 1-\frac{\alpha}{K} \right) - M_1.$$
For the comparison, take $M_1=(1-\alpha)N$, and we have
$$\overline{M'}=\alpha+(1-\alpha)  N  \quad \text{and} \quad R'=2N \left( 1-\frac{\alpha}{K} \right) - (1-\alpha)N.$$
Now, we have
\begin{align*}
R-R' &= \alpha N \left( 1- \frac{1}{K}\right) + N - 2N \left( 1-\frac{\alpha}{K} \right) - (1-\alpha)N  = \frac{\alpha N}{K} >0.
\end{align*}
Therefore, this scheme performs better than scheme given in Section \ref{scheme1} for the hierarchical system when there is only one mirror.

\section{Comparison with the LZX scheme} \label{scheme2compare}
The LZX scheme given in \cite{liu2021intension} considers the two-layered hierarchical network with a single mirror with two users and is shown to be optimal in an average sense. In \cite{liu2021intension}, the scheme for the case $2M_2 \leq N$ is given in detail, and then the general idea is provided for the case $2M_2 \geq N$. For $2M_2 \leq N$, the rates of the LZX scheme are given in the following cases depending on the size of $M_1$. Considering $a=\frac{M_2}{N}$ and $b=1-\frac{2M_2}{N}$,

\begin{enumerate}
\item If $M_1 \leq Nb$,
$$R_1=\frac{1}{N^2} \left((2N-1)(N-M_1)-(3N-2)M_2\right),$$
$$R_2=\frac{1}{N^2} \left((2N-1)N-(3N-2)M_2\right).$$
\item If $Nb < M_1 \leq Nb+Na$,
$$R_1=1-\frac{M_1+M_2}{N},$$
$$R_2=\frac{1}{N^2} \left(N(N-M_2)+(N-1)M_1\right).$$
\item If $Nb+Na < M_1 < Nb+(2N-1)a$,
$$R_1=0, \ R_2=\frac{1}{N^2} \left((3N-1)(N-M_2)-NM_1\right).$$
\item If $Nb+(2N-1)a \leq M_1$,
$$R_1=0, \ R_2=\frac{1}{N^2} \left(N(2N-1)-(3N-2)M_2\right).$$
\end{enumerate} 
 In a similar setup of a single mirror with two users, our scheme for the single mirror case achieves the lowest rate as the LZX scheme only for the case when $\alpha=1$. If $\alpha=1$, then the parameters of our scheme for the single mirror case are $M_2=\frac{1}{2}, M_1=0, R_1=1, R_2=1$ and the composite rate $\overline{R}=2$. The LZX scheme achieves the same composite rate $\overline{R}=2$ for Memory $M_2=\frac{1}{2}$ and $M_1=0$. For other memory points $(M_1, M_2)$, the LZX scheme performs better than our scheme. However, the LZX scheme works only for two users, whereas our scheme for the single mirror case works for any number of users. 

\section{Conclusion}\label{conclu}

In this study, we deal with the two-layer hierarchical coded caching problem by introducing a coded placement scheme. Our scheme significantly lowers transmission rates when compared to existing approaches, and for the comparison, we consider two parameters: global memory and composite rate. Our scheme is designed to accommodate scenarios where the number of users is equal to or greater than the number of files. Also, we gave an alternate scheme for the special case of one mirror, which improved the rate of the existing scheme for one mirror case. Our scheme is not achieving the optimal rate for the case of one mirror and two users, which is achieved by the LZX scheme. However, the LZX scheme is not applicable to multiple users. Therefore, finding another scheme for the case of one mirror and multiple users with a better rate is an interesting problem to consider.

\section*{Acknowledgement}
This work was supported partly by the Science and Engineering Research Board (SERB) of the Department of Science and Technology (DST), Government of India, through J.C Bose National Fellowship to Prof. B. Sundar Rajan, by the Ministry of Human Resource Development (MHRD), Government of India, through Prime Minister’s Research Fellowship (PMRF) to Rajlaxmi Pandey, and through Indian Institute of Science (IISc) CV Raman Postdoc Fellowship to Charul Rajput.

\section*{Appendix A}

\subsection{Proof of Theorem \ref{thm_region}}
\begin{proof}
The KNMD scheme \cite{karamchandani2016hierarchical} gives three regions of memory on the basis of $M_{1}+M_{2} K_{2} \gtreqqless N$. For $N=K$, this can be written as $
M_1+M_2 K_2  \gtreqqless  K.$
We calculate $M_{1}+M_{2} K_{2}$ for $t=K_{2}$ for our scheme as:
\begin{align*}
 M_1+M_2 K_2 = \frac{\alpha}{K_1} + (1-\alpha) \frac{K}{{K \choose K_2}} + \frac{(1-\alpha)}{{K \choose K_2}} \\ \left({{K-1} \choose {K_2-1}} - 1\right) K K_2 
= \alpha\left(\frac{1}{K_1}-A\right)  + A, 
\end{align*} 
where 
\begin{align*}
    A=\frac{K}{\left(\begin{array}{l}
K \\
K_2
\end{array}\right)}+\left[\frac{\left(\begin{array}{l}
K-1 \\
K_2-1
\end{array}\right)-1}{\left(\begin{array}{l}
K \\
K_2
\end{array}\right)}\right]K K_2.
\end{align*}
Simplifying this further, we get 
$
   A= K_2^2- \frac{K(K_2-1)}{{K \choose {K_2}}}.
$
Therefore, we have
\begin{align*}
M_1+M_2K_2 \geq K \ &\iff \ \alpha \leq \frac{A-K}{A-\frac{1}{K_1}}, \\
M_1+M_2K_2 < K \ &\iff \ \alpha > \frac{A-K}{A-\frac{1}{K_1}}. 
\end{align*}
Now, we check condition for Region I and III separately as follows: \\
For Region I:
\begin{align*}
& 0 \leq M_1 \leq N / 4 \\
 & 0 \leq \frac{\alpha}{K_1}+(1-\alpha) \frac{K}{\left(\begin{array}{c}
K \\
K_2\end{array}\right)} \leq K / 4.
\end{align*}
After simplifying it, we get
\begin{align*}
 \alpha \leq  \frac{KK_{1}}{4}  \left[\frac{\left(\begin{array}{l}
K \\
K_2
\end{array}\right)-4}{\left(\begin{array}{l}
K \\
K_2
\end{array}\right)-KK_{1}}\right].
\end{align*}

Let $B = \frac{KK_{1}}{4} \left[\frac{{K \choose K_2}-4}{{K \choose K_2}-KK_{1}}\right]$.
Therefore, we can redefine regions of the KNMD schemes for the proposed scheme as follows:

I) $\alpha \leq  \min{ \left(\frac{A-K}{A-\frac{1}{K_1}}, B\right)}$,

II) $\alpha \geq \frac{A-K}{A-\frac{1}{K_1}}$,

III) $B \leq \alpha \leq   \frac{A-K}{A-\frac{1}{K_1}}$.

\end{proof}

\subsection{Proof of Theorem \ref{thm2}}
The following lemma is used in the proof of Theorem \ref{thm2}.
\begin{lem}\label{lem1}
For integers $K_1 \geq 2$ and $K_2 \geq 2$, except $K_1=K_2=2$, we have ${K \choose K_2} > K K_2$, where $K=K_1K_2$.
\end{lem}

\begin{proof}
We know,$${K \choose K_{2}} = \frac{{K!}}{{K_{2}!(K-K_{2})!}}$$ 
\begin{align*}
&= \frac{{(K)(K-1)\cdots(K-K_{2}+1)}}{{K_{2}(K_{2}-1)\cdots 2\cdot1}} \\
&=K\left[\frac {K-1}{K_{2}}\frac{K-2}{K_{2}-1} \cdots\frac{K-K_{2}+1}{2}\right] \\
  &=K\prod_{i=1}^{K_{2}-1} \frac{K-i}{K_{2}-i+1}
\end{align*}
For $K_2=2$, we can easily check that ${K \choose K_{2}} > KK_2$ if $K_1 \geq 3$. Now for $K_2 \geq 3$, we have
$${K \choose K_{2}} =KK_{2}\left(\frac{K_{1}-1+\frac{1}{K_{2}}}{2}\right)\prod_{i=1}^{K_{2}-2} \frac{K-i}{K_{2}-i+1}$$
We know, $K \geq K_{2}+1$, which implies that $ K-i \geq K_{2}-i+1$ and
$ \frac{K-i}{K_{2}-i+1} \geq 1.$
Hence, the term inside the product is always greater than one.
The term $\left(\frac{K_{1}-1+\frac{1}{K_{2}}}{2}\right)$ is always greater than one for $K_{1}\geq 3$. Therefore, we have 
$${K \choose K_2} \geq KK_{2}\left(\frac{K_{1}-1+\frac{1}{K_{2}}}{2}\right)>  KK_{2}.$$
Again for $K_1=2$ and $K_2 \geq 3$, we can easily check that ${K \choose K_2} > KK_{2}$.

\end{proof}

\begin{proof}[\textbf{Proof of Theorem \ref{thm2}}]
For integers $K_1 \geq 2$ and $K_2 \geq 2$ such that $K_1 > K_2$, we will show that $\frac{A-K}{A-\frac{1}{K_1}} < 0$. Since $\alpha>0$, it will imply that $\alpha>\frac{A-K}{A-\frac{1}{K_1}}$ and the proposed scheme lies in Region II. 
We have 
$$A = K_2^2- \frac{K(K_2-1)}{{K \choose {K_2}}}= K_2^2- \frac{KK_2(1- \frac{1}{K_2})}{{K \choose {K_2}}} .$$ 
From Lemma \ref{lem1}, we have ${K \choose {K_2}} > KK_2$ for $K_1 \geq 2$ and $K_2 \geq 2$ except $K_1=K_2=2$. Therefore, we have
$$A > K_2^2 +1-\frac{1}{K_2},$$
and 
\begin{equation}\label{A1}
A - \frac{1}{K_1}> K_2^2 +1-\frac{1}{K_2}- \frac{1}{K_1}>0.
\end{equation}
Also, we can check for $K_1=K_2=2$, we have $A - \frac{1}{K_1} >0.$ Further, we have 
\begin{align}\label{A2}
A-K&=K_2^2- \frac{K(K_2-1)}{{K \choose {K_2}}} -K \nonumber \\
& = \frac{(K_2^2-K) {K\choose K_2}-K(K_2-1)}{{K\choose K_2}} \nonumber\\
& <0 \ \left( \text{as} \ K_1 > K_2 \ \text{and} \ K_2^2 - K <0 \right).
\end{align}
From \eqref{A1} and \eqref{A2}, we have $\frac{A-K}{A-\frac{1}{K_1}} < 0$.
\end{proof}

\subsection{Proof of Theorem \ref{thm3}}}

The following result can be directly obtained by Lemma \ref{lem1}.

\begin{lem}\label{cor1}
For integers $K_1 \geq 2$ and $K_2 \geq 2$ except $K_1=K_2=2$, if ${K \choose K_2} \leq K K_1$ then $K_1>K_2$,  where $K=K_1K_2$.
\end{lem}  

The contrapositive statement of the Lemma \ref{cor1} is as follows, and this result will be used in the proof of Theorem \ref{thm3}.
\begin{lem}\label{cor2}
For integers $K_1 \geq 3$ and $K_2 \geq 2$ except $K_1=K_2=2$, if $K_1 \leq K_2$ then ${K \choose K_2} > K K_1$,  where $K=K_1K_2$.
\end{lem}  

\begin{proof}[\textbf{Proof of Theorem \ref{thm3}}]
For integers $K_1 \geq 2$ and $K_2 \geq 2$ such that $K_1 \leq K_2$, we will show that $B>1$. Since $\alpha<1$, it will imply that $\alpha<B$ and the proposed scheme will never lie in Region III. We have
\begin{align*}
B-1 & = \frac{KK_1}{4} \left( \frac{{K \choose K_2}-4}{{K \choose K_2}-KK_1} \right) -1 \\
& = \frac{\frac{KK_1}{4}  \left( {K \choose K_2}-4 \right) - {K \choose K_2} + KK_1 }{{K \choose K_2}-KK_1} \\
& = \frac{{K \choose K_2}\left( \frac{KK_1-4}{4} \right) }{{K \choose K_2}-KK_1}
\end{align*}
From Lemma \ref{cor2}, we have ${K \choose {K_2}} > KK_1$ for $K_1 \geq 2$ and $K_2 \geq 2$ except $K_1=K_2=2$ if $K_1 \leq K_2$. Therefore, we have
$$B-1 > 0 \ \left(\text{as} \ {K \choose {K_2}} - KK_1 >0 \right).$$
\end{proof}


\end{document}